\documentclass[leqno,11pt]{article}

\usepackage{subcaption}
\usepackage{xspace}
\usepackage{amsmath}
\usepackage{tikz-cd}
\usepackage{mathtools}
\usepackage{multicol}

\usepackage[ruled,linesnumbered,lined,commentsnumbered]{algorithm2e}

\usepackage[english]{babel}
\usepackage{amsmath,amsthm,amssymb,mathrsfs} 
\usepackage{ae,aecompl}
\usepackage{times}
\usepackage[pagebackref,hypertexnames=false]{hyperref}
\usepackage{fancybox,fancyhdr,graphics,epsfig}
\usepackage{textcomp}
\usepackage{authblk}

\newtheorem{lem}{Lemma}[section]
\newtheorem{thm}[lem]{Theorem}
\newtheorem{prop}[lem]{Proposition}
\newtheorem{cor}[lem]{Corollary}

\newtheorem{remark}[lem]{Remark}

\theoremstyle{definition}
\newtheorem{definition}[lem]{Definition}

 \usepackage{amsmath,amsthm,bbm,amssymb}

\DeclareMathOperator{\PH}{PH}
\DeclareMathOperator{\Ker}{Ker}
\DeclareMathOperator{\Imm}{Im}
\DeclareMathOperator{\bth}{birth}
\DeclareMathOperator{\dth}{death}
\DeclareMathOperator{\vor}{Vor}

\DeclareMathOperator{\rank}{rank} 
\DeclareMathOperator{\Int}{Int}

\DeclareMathOperator{\Cech}{\check{C}ech}
\DeclareMathOperator{\Rips}{Rips}
\DeclareMathOperator{\Vol}{Vol}
\DeclareMathOperator{\spn}{span}

\DeclareMathOperator{\diam}{diam}

\newcommand{\coA}{\cA^\mathrm{co}}
\newcommand{\cech}{\v{C}ech\xspace}
\newcommand{\bcup}{\bigcup\limits}
\newcommand{\bcap}{\bigcap\limits}

\newcommand{\Rd}{\mathbb{R}^{d}}
\newcommand{\Rdd}{\mathbb{R}^{d+1}}

\newcommand{\Pn}{\mathcal{P}_n}
\newcommand{\bs}{\boldsymbol}

\newcommand{\ind}{\mathbbm{1}}
\newcommand{\E}{ \mathbb{E}}

\newcommand\R{{\mathbb{R}}}

\def\cX{{\cal{X}}}
\def\cY{{\cal{Y}}}

\def\cC{{\cal{C}}}

\def\cA{{\cal{A}}}
\def\cD{{\cal{D}}}
\def\cN{{\cal{N}}}
\def\cP{{\cal{P}}}

\begin{document}

\title{A Coupled Alpha Complex}

\author[1]{Yohai Reani \thanks{syohai@campus.technion.ac.il}}
\author[1]{Omer Bobrowski \thanks{omer@ee.technion.ac.il}}
\affil{Viterbi Faculty of Electrical Engineering\\Technion - Israel Institute of Technology}

\maketitle

\begin{abstract}
The alpha complex is a subset of the Delaunay triangulation and is often used in computational geometry and topology. 
One of the main drawbacks of using the alpha complex is that it is non-monotone, in the sense that if ${\cal X}\subset{\cal X}'$ it is not necessarily (and generically not) the case that the corresponding alpha complexes satisfy ${\cal A}_r({\cal X})\subset{\cal A}_r({\cal X}')$. 
The lack of monotonicity may introduce significant computational costs when using the alpha complex, and in some cases even render it unusable. 
In this work we present a new construction based on the alpha complex, that is homotopy equivalent to the alpha complex while maintaining monotonicity. 
We provide the formal definitions and algorithms required to construct this complex, and to compute its homology. In addition, we analyze the size of this complex in order to argue that it is not significantly more costly to use than the standard alpha complex.
\end{abstract}

\section{Introduction}

The alpha complex \cite{edelsbrunner_alpha_shapes} is a parametrized triangulation constructed over point clouds. It is widely used in computer graphics \cite{krone_biomolecular_2016}, computational geometry \cite{liang_macromolecules_1998, liang_anatomy_1998}, topological data analysis (TDA) \cite{pranav_cosmic_web_2016}, and other fields.
Given a point cloud $\cX = \{x_1,...,x_n\}\subset\Rd$, the alpha complex $\cA_r(\cX)$ is a $d$-dimensional simplicial complex consisting of a subset of the faces in the Delaunay triangulation of $\cX$. In TDA, one of its main uses is as a substitute for the \cech complex $\Cech_r(\cX)$ (the nerve of the balls of radius $r$ centered at $\cX$), justified by the fact that the alpha and the \cech complexes are homotopy equivalent. While the \cech complex is highly useful to develop the theory and intuition in TDA (especially in probabilistic analysis \cite{bobrowski2019homological,kahle2011random}), using the alpha complex in applications is significantly more efficient computationally. The \cech complex contains $O(n^{k+1})$ many $k$-simplexes
, and those can appear in any dimension. On the other hand, the alpha complex contains simplexes only up to dimension $d$, and it can be shown \cite{seidel_upperbound_1995} that there are at most $O(n^{\lceil d/2\rceil})$ many of them. Moreover, for generic random point clouds, it can be shown \cite{edelsbrunner_expected_2017} that the alpha complex has only $O(n)$ many simplexes. Since computing homology or persistent homology, for example, requires cubical time in the number of simplexes, such a difference in the complex size can be crucial.

One of the main drawbacks of the alpha complex is the following. For any finite $\cX\subset\cY\subset \R^d$ we have a natural inclusion $\Cech_r(\cX)\subset \Cech_r(\cY)$. However, the same is not true in general for the alpha complex. In other words, adding new points to an existing alpha complex, requires us to re-calculate the entire complex.
There are various scenarios where the lack of such an inclusion can prevent us from using alpha complexes. For example:

\begin{enumerate}
\item \emph{Computing zigzag-persistence} \cite{carlsson_zigzag_2009}. As opposed to the standard persistent homology that is (commonly) computed over filtrations, in zigzag persistence the inclusion relations may go in different ways. For example, suppose that we have a sequence of point clouds $\cX_1,\cX_2\ldots$, with no inclusion relation, and we wish to find cycles that persist throughout this sequence. Using the \cech complex, we can take the sequence
\[
    \Cech_r(\cX_1) \hookrightarrow \Cech_r(\cX_1\cup \cX_2) \hookleftarrow \Cech_r(\cX_2)\hookrightarrow \Cech_r(\cX_2\cup \cX_3) \hookleftarrow \Cech_r(\cX_3) \hookrightarrow \cdots,
\]
and compute its zigzag persistence barcode, for example. However, we are currently not able to do so using alpha complexes, and the computational implications are substantial.

\item \emph{Cycle registration in persistent homology} \cite{reani2021cycle}. We recently presented a new framework for identifying matching persistent cycles between pairs of simplicial filtrations. For example, suppose that we have two point clouds $\cX,\cY\subset\R^d$. We can use the \cech filtration to compute two persistent modules $\PH_k(\cX)$ and $\PH_k(\cY)$. Our goal is to find pairs of persistent-cycles (classes) in $\PH_k(\cX)$ and $\PH_k(\cY)$ that represent  the ``same topological phenomenon" (which we define rigorously in \cite{reani2021cycle}).
Our solution heavily relies on the inclusions $\Cech_r(\cX)\hookrightarrow \Cech_r(\cX\cup \cY) \hookleftarrow \Cech_r(\cY)$, and specifically on the images of the induced maps in homology. Here as well, we cannot use the alpha complex. At the same time, using the \cech complex (or the Vietoris-Rips) becomes infeasible for rather small sample sizes.
\end{enumerate}

In order to resolve this fundamental issue, we present here a new construction we call the \emph{coupled alpha complex} and denote by $\coA_r(\cX,\cY)$. 
This is a new ``hybrid'' complex defined over pairs of finite point clouds $\cX,\cY\subset\Rd$. The key properties of this new complex are: (a) It is homotopy equivalent to the alpha complex $\cA_r(\cX\cup\cY)$; (b) It satisfies the desired inclusions that the alpha complex misses, i.e., $\cA_r(\cX)\hookrightarrow\coA_r(\cX,\cY)$ and $\cA_r(\cY)\hookrightarrow\coA_r(\cX,\cY)$;  (c) The total number of simplexes in this complex is $O(n^{\lceil (d+1)/2 \rceil})$ (still smaller compared to \cech), and for random point clouds we can show that the expected size goes down to $O(n)$.

\paragraph{Related work.} 
In \cite{bauer_morse_2017} the authors introduce the  \emph{selective Delaunay complex}, defined for a subset $E\subset X\subset\Rd$ of excluded points. This complex is a subset of simplexes $Q\subset X$ for which there exists a sphere that includes $Q$ and its interior excludes $E$.
The alpha and \cech complexes are extremal cases obtained by choosing $E=X$ and $E=\emptyset$, respectively.
In this paper we suggest a similar yet different construction. 
While in a selective Delaunay complex all the simplexes are induced by a fixed subset $E$, in our case, we have two distinct sets $\cX$ and $\cY$ of excluding points. Specifically, it is a subset of simplexes $Q\in \cX\cup\cY$ for which there exist two concentric spheres, one that includes $Q\cap\cX$ and its interior excludes $\cX$, and one that includes $Q\cap\cY$ and its interior excludes $\cY$.
In addition, while in general the selective Delaunay complex includes only $\cA_r(\cX\cup\cY$), ours also includes $\cA_r(\cX)$ and $\cA_r(\cY)$.

In \cite{blaser_relative_2020} the authors define  the \emph{Relative Delaunay-\v{C}ech complex} which is a complex designed for computing the persistent homology of $X$ relative to $A$, where $A\subset X\subset\Rd$. This complex is a subset of the Delaunay triangulation of $A\times\{0\}\cup (X\setminus A)\times\{1\}\subset\Rdd$,
which is similar in spirit to the way we construct the coupled alpha complex in Section \ref{sec:construct}.
One of the main differences between these constructions is the filtration values assigned to the simplexes. In the relative Delaunay-\v{C}ech complex, the filtration value of each $Q\subset A$ is $0$, while for $Q\not\in A$ it is the radius its minimal bounding sphere.
On the other hand, in the coupled alpha complex the filtration values are computed in a top down fashion, with no distinction between the sets $\cX$ and $\cY$. More generally, while the relative Delaunay-\v{C}ech complex is designed specifically for computing relative persistent homology, the coupled alpha complex is a general-purpose tool that serves as a ``bridge'' between arbitrary alpha complexes. Moreover, one can obtain the relative Delaunay-\v{C}ech complex (up to homotopy equivalence) from the coupled alpha complex, by assigning some of the simplexes with a filtration value of zero. Finally, the formalism we present here, simplifies the derivation of a probabilistic upper bound on the complex size, presented in Section \ref{sec:num_simp}.

\paragraph{Paper outline.} 
In Section \ref{sec:prelims} we give a brief introduction for the terms and the structures discussed in this paper.
In Sections \ref{sec:coup_alpha} and \ref{sec:construct} we introduce the main contribution of this paper -- the \emph{coupled alpha complex}. We present the formal definition and provide a two-step computation scheme. 
Finally, in Section \ref{sec:num_simp} we provide a probabilistic upper  bound for the expected size of the coupled alpha complex, in the case where the points cloud is generated at random.

\section{Preliminaries}
\label{sec:prelims}

In this section we give a brief introduction to simplicial complexes, homology and persistent homology. For more details  \cite{edelsbrunner_topological_2002,edelsbrunner_persistent_2008,edelsbrunner_persistent_2014,hatcher_algebraic_2002,zomorodian_computing_2005}. 

\subsection{Simplicial homology}

An \emph{abstract simplicial complex} over a set $S$, is a collection of finite subsets $K$ that is closed under inclusion, i.e., if $P\in K$ and $Q\subset P$, then, $Q\in K$. The elements of $K$ are called \emph{simplexes} and their dimension is determined by their size minus one. For $Q\subset P\in K$, we say that $Q$ is a \emph{face} of $P$, and $P$ is a \emph{co-face} of $Q$ of \emph{co-dimension} $d$, where $d=\dim(P)-\dim(Q)$.  

\paragraph{Homology.}
Homology is a topological-algebraic structure that describes the shape of a topological space by its connected components, holes, cavities, and generally  $k$-dimensional cycles (see Figure \ref{fig:Hom_example}). 
Loosely speaking, given a topological space $X$, $H_0(X)$ is an abelian group generated by elements that correspond to the connected components of $X$; similarly $H_1(X)$ is generated by ``holes'' in $X$; $H_2(X)$ is generated by the ``cavities'' or ``bubbles'' in $X$. Generally, we can define the group $H_k(X)$ generated by the non-trivial $k$-dimensional cycles of $X$. A $k$-dimensional cycle can be thought of as the boundary of a $(k+1)$-dimensional object (i.e.~with the interior excluded).

\begin{figure}[h]
    \centering
    \includegraphics[scale=0.6]{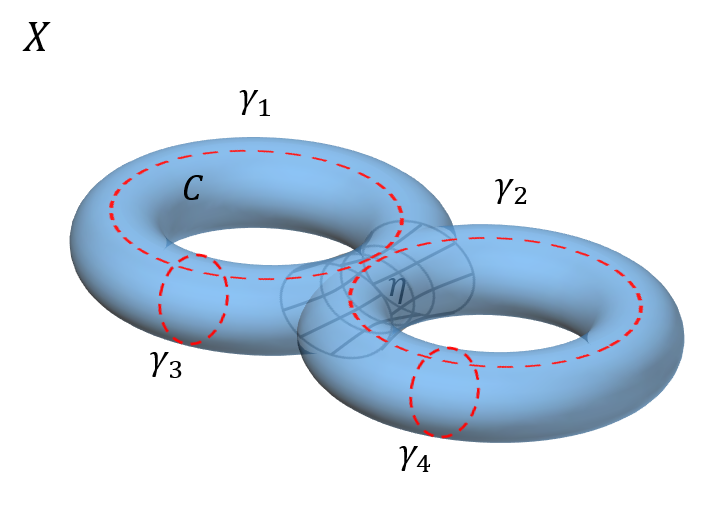}
    \caption{Homology example. The space here is 2-dimensional surface of genus 2. This manifold is composed of a single connected component, denoted by $C$, has four $1$-dimensional holes denoted by $\gamma_1, \gamma_2, \gamma_3$ and $\gamma_4$ (emphasized by the red dashed circles), and one $2$-dimensional hole, denoted by $\eta$, which is the entire surface that encloses the air-pocket inside. Hence, loosely speaking the homology groups of $X$ are given by $H_0(X)=\spn\{C\}$, $H_1(X)=\spn\{\gamma_1,\gamma_2,\gamma_3,\gamma_4\}$, $H_2(X)=\spn\{\eta\}$ and $H_k(X)=0$ for $k\geq3$.}
    \label{fig:Hom_example}
\end{figure}

Formally, let $X$ be a simplicial complex. The $k$-dimensional chain group $C_k(X)$ is a free abelian group generated by the $k$-dimensional simplexes in $X$. 
In this article we will use $\mathbb{Z}_2$ coefficients, and therefore $C_k(X)$ is a vector space. The elements of $C_k(X)$ are formal sums of $k$-simplexes called \emph{chains}.
The \emph{boundary homomorphism} $\partial_k:C_k(X)\to C_{k-1}(X)$ is defined as follows. If $\sigma$ is a chain representing a single $k$-simplex, then
$
	\partial_k(\sigma) = \sum_{\tau < \sigma}\tau$,
where $\tau < \sigma$ denotes that $\tau$ is a $(k-1)$-dimensional face of $\sigma$.
For general chains $\gamma=\sum_{i}\sigma_i\in C_k(X)$, $\partial_k$ extends linearly, i.e.~
$
    \partial_k(\gamma)=\sum_{i}\partial_k(\sigma_i)$.
It can be shown that $\partial_{k-1}\circ\partial_k \equiv 0$ for every $k>0$, and the sequence 
\[
\cdots\rightarrow C_{k+1}(X)\overset{\partial_{k+1}}{\rightarrow} C_k(X) \overset{\partial_k}{\rightarrow} C_{k-1}\rightarrow\cdots
\]
is known as a \emph{chain complex}.
Next, we define the subgroups
\[
Z_k(X)=\Ker(\partial_k(X)),\qquad
B_k(X)=\Imm(\partial_{k+1}(X)),
\]
so that $B_k(X)\subset Z_k(X)$. The group $Z_k(X)$ is known as the \emph{$k$-cycle group} (i.e.~chains whose boundary is zero) and $B_k(X)$ as the \emph{$k$-boundary group} (i.e.~$k$-cycles that are boundaries of $(k+1)$-dimensional chains).
The $k$-th \emph{homology group} is then defined as the quotient group,
$$H_k(X)= Z_k(X) / B_k(X).$$
In other words, the $k$-th homology group $H_k(X)$ consists equivalence classes of $k$-dimensional cycles who differ only by a boundary (called \emph{homological} cycles).
The ranks of the homology groups, called the \emph{Betti numbers}, are denoted
$\beta_k = \rank(H_k)$.

As mentioned earlier, intuitively speaking, the generators (or basis) of $H_0(X)$ correspond to the connected components of $X$, $H_1(X)$ corresponds to the holes in $X$, and $H_2(X)$ are the cavities. The definitions provided above are for simplicial homology, while other notions of homology groups can be defined for a much larger classes of topological spaces (see \cite{hatcher_algebraic_2002}). The intuition, however, is similar. 

In addition to homology, throughout the paper we will also use the following two terms which can be defined for simplicial homology as well as the more general notions of homology.

\paragraph{Simplicial maps and induced homomorphisms.} Let $X$ and $Y$ be simplicial complexes and let $f:X\rightarrow Y$ be a simplicial map, i.e., $f([v_0,...,v_k])=[f(v_0),...,f(v_k)]\in Y$  for $[v_0,...,v_k]\in X$. 
Then homology theory provides a sequence of induced functions denoted $f_* :H_k(X)\rightarrow H_k(Y)$, that map $k$-cycles in $X$ to $k$-cycles in $Y$.

\paragraph{Homotopy Equivalence.} This is a notion of similarity between spaces that is weaker than homeomorphism.
Loosely speaking, two topological spaces $X$ and $Y$ are \emph{homotopy equivalent}, denoted $X\simeq Y$, if one can be continuously deformed into the other. 
In particular, homotopy equivalence between spaces implies similar homology, i.e.~if $X\simeq Y$, then $H_k(X)\cong H_k(Y)$ for all $k\geq 0$.

\subsection{Geometric complexes}
\label{sec:geo_comp}
Simplicial complexes are the fundamental building blocks in many TDA methods, where they are used for approximating geometric shapes using discrete structures. 
In this section we present a few special types of geometric complexes commonly used in TDA.

\begin{definition}[\v{C}ech Complex]\label{def:cech}
Let $\cX$ be a finite set of points in a metric space. The \emph{\v{C}ech} complex of $\cX$ with radius $r$, denoted by $\Cech_r(\cX)$, is an abstract simplicial complex, constructed using the intersections of balls around $\cX$,
$$\Cech_r(\cX) :=  \big\{ Q\subset\cX : \bcap_{x\in Q} B_r(x)\neq\emptyset \big\},$$
where $B_r(x)$ is a ball of radius $r$ centered at $x$. 
\end{definition}

\begin{definition}[Vietoris-Rips Complex]\label{def:rips}
Let $\cX$ be a finite set of points in a metric space. The \emph{Vietoris-Rips} complex of $\cX$ with radius $r$, denoted by $\Rips_r(\cX)$, is an abstract simplicial complex, constructed  by pairwise intersections of balls,
$$\Rips_r(\cX) :=  \big\{Q \subset \cX : \forall x, x'\in Q\text{, }B_r(x)\cap B_r(x')\neq\emptyset\big\}.$$
In other words, all the points in $Q$ are within less than distance $2r$ from each other.
\end{definition}

\begin{figure}
     \centering
     \begin{subfigure}[b]{0.3\textwidth}
         \centering
         \includegraphics[width=\textwidth]{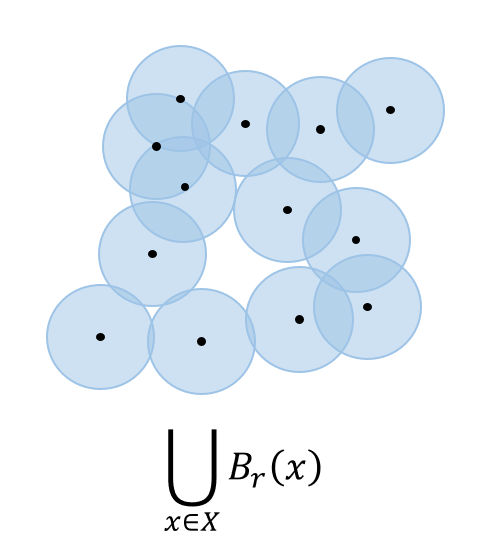}
        \caption{}
     \end{subfigure}
     \begin{subfigure}[b]{0.3\textwidth}
         \centering
         \includegraphics[width=\textwidth]{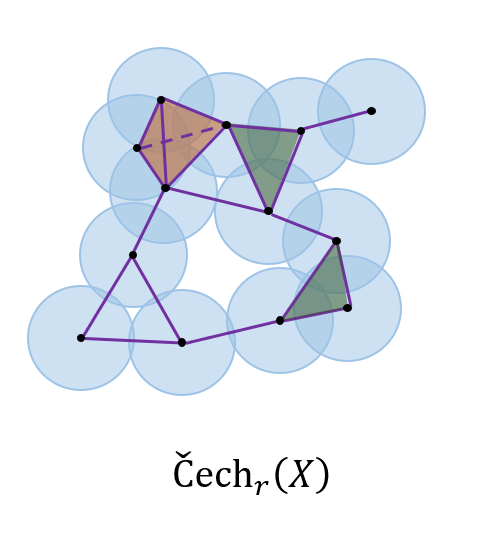}
         \caption{}
     \end{subfigure}
     \begin{subfigure}[b]{0.3\textwidth}
         \centering
         \includegraphics[width=\textwidth]{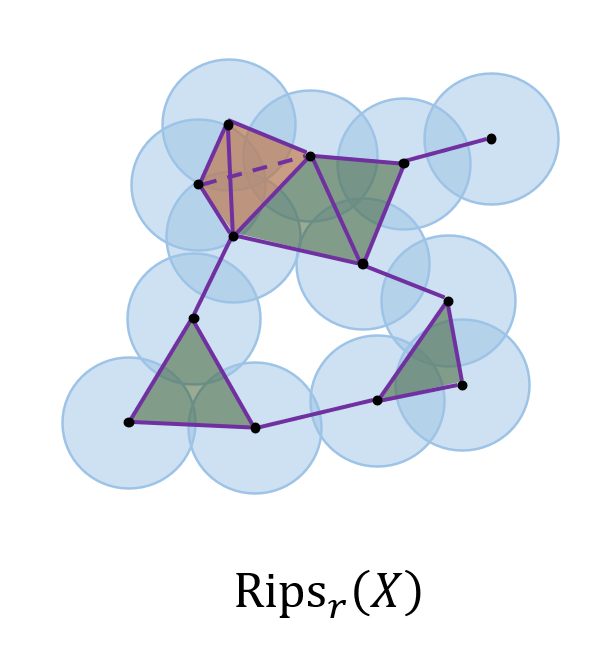}
         \caption{}
     \end{subfigure}
     \caption{(a) $\bigcup_{x\in X}B_r(x)$ - the ball cover of $X$ with radius $r$. (b) $\Cech_r(X)$ - the \v{C}ech complex of $X$ with radius $r$. (c) $\Rips_r(X)$ - the Vietoris-Rips complex of $X$ with radius $r$.}
        \label{fig:Cech_Rips}
\end{figure}

The \v{C}ech and Rips complexes are the most extensively studied complexes in TDA.
The Rips complex is commonly used in applications (see for example \cite{de_silva_homological_2007})  due to its simple definition that depends on pairwise distances only. The Rips complex can be also viewed as an approximation for the \v{C}ech complex by the following relation \cite{de_silva_coverage_2007,edelsbrunner_computational_2010},
$$\Rips_r(\cX)\hookrightarrow\Cech_{\sqrt{2}r}(\cX)\hookrightarrow\Rips_{\sqrt{2}r}(\cX).$$
The construction of the \v{C}ech complex is a bit more intricate, hence it is slightly less popular in applications. However, the \v{C}ech complex plays a central role in many theoretical results, especially in the random setting (see for example \cite{kahle2011random,yogeshwaran2017random,bobrowski2019homological,auffinger2020topologies,skraba2020homological}). This largely due to the fact that the \cech complex is homotopy equivalent to the ball cover inducing it due to the Nerve Lemma which we state in the following.

\begin{definition}[Nerve of a covering]
Let $X$ be a topological space and let $\mathcal{U}=\{U_i\}_{i\in I}$ be a cover of $X$. The Nerve of $\mathcal{U}$, denoted by $\cN(\mathcal{U})$, consists of all finite subsets $J\subset I$ such that, 
$$\bcap_{i\in J}U_i\neq\emptyset.$$
\end{definition}
Note that by definition $\cN(\mathcal{U})$ is an abstract simplicial complex.

\begin{lem}[Nerve Lemma \cite{borsuk_imbedding_1948}]\label{lemma:nerve}
Let $X$ be a topological space and let $\mathcal{U}=\{U_i\}_{i\in I}$ be a good cover of $X$, i.e.~for every $J\subset I$ the set $\bcap_{i\in J}U_i$ is either contractible or empty. Then, the nerve $N(\mathcal{U})$ is homotopy equivalent to $\bcup_{i\in I}U_i$.
\end{lem}

A direct result of the Nerve Lemma is the following corollary.
\begin{cor} Let $\cX\subset\Rd$ be a finite set of points. Then,
$$\Cech_r(\cX)\simeq \bigcup_{x\in\cX}B_r(x).$$
\label{cor: cech_nerve}
\end{cor}

Corollary \ref{cor: cech_nerve} implies that the homology groups of $\Cech_r(\cX)$ and $\bigcup_{x\in\cX}B_r(x)$ are isomorphic. Hence, they can be used interchangeably when trying to prove a result concerning their homotopy type or homology groups.

The next complex we discuss, the \emph{alpha complex}, serves as a basis for the construction of the \emph{coupled alpha complex} introduced in Section \ref{sec:coup_alpha}. The alpha complex is homotopy equivalent to the \v{C}ech complex, but with much fewer simplexes. 
Let $\cX$ be a finite set of points in a metric space $(M,d)$. 
The \emph{Voronoi cell} of $x\in\cX$ with respect to $\cX$ is defined as 
$$\vor(x,\cX):= \left\{z\in M : d(x,z)\leq d(x',z),\ \forall x'\in\cX\right\}.$$
In addition, we define the \emph{Voronoi ball} of $x$ with respect to $\cX$, as
\begin{equation}\label{eqn:vor}
\vor_r(x,\cX) :=B_r(x)\cap \vor(x,\cX).
\end{equation}

\begin{definition}[Alpha complex]\label{def:alpha}
Let $\cX$ be a finite set of points in a metric space. The Alpha Complex of $\cX$ with parameter $r$, denoted by $\cA_r(\cX)$, is defined as the nerve of all the Voronoi balls, i.e.
$$\cA_r(\cX) :=  \big\{ Q\subset\cX : \bcap_{x\in Q} \vor_r(x,\cX)\neq\emptyset \big\}.$$
\end{definition}

By the Nerve Lemma, we have that $\cA_r(\cX) \simeq \Cech_r(\cX)$, and in particular they have the same homology.

Throughout this article we will assume that a given point set is in \emph{general position}, defined as follows.

\begin{definition}[General Position] \label{def:gen_pos}
A finite set $P\subset\Rd$ ($|P|\geq d+1$) is said to be in general position, if for every $Q\subseteq P$ of size $d+1$,
\begin{enumerate}
    \item The points of $Q$ do not lie on a $(d-1)$-dimensional flat.
    \item No point of $P\setminus Q$ lies on the circumsphere of $Q$.
\end{enumerate}
\end{definition}
In this case, the alpha complex $\cA_r(\cX)$ can be realized as a (geometric) simplicial complex embedded in $\mathbb{R}^d$ (as opposed to the \v{C}ech complex), i.e.~it also satisfies the following geometric condition,
\[
Q,Q'\in \cA_r(\cX) \quad \Longrightarrow\quad Q\cap Q'\in \cA_r(\cX).
\]
In other words, if two embedded simplexes intersect, they must intersect along a common face.
In addition, for $r$ large enough, $\cA_r(\cX)$ becomes identical to the \emph{Delaunay triangulation}, defined as follows.
\begin{definition} [Delaunay Triangulation]
Let $\cX\subset\Rd$ be a finite set. The Delaunay Triangulation of $\cX$, denoted by $\mathcal{D}(\cX)$, is a triangulation of $\cX$ such that the circumsphere of each $d$-simplex in the triangulation does not contain any point of $\cX$.
\end{definition}
Note that for sets in general position the Delaunay triangulation is unique (see \cite{boissonnat2018geometric}).

\paragraph{Alpha complex computation.} While algorithms for constructing the \cech and Rips complexes are  derived directly from their definitions, an algorithm for the Alpha complex is derived from its definition in a rather dual way, based on the following proposition (see Figure \ref{fig:alpha_comp}). 
\begin{prop}[\cite{edelsbrunner_alpha_shapes}] 
\label{prop: Del_eq_alpha}
Let $\cX\subset\Rd$ be a finite set in general position. Then, 
$$\mathcal{D}(\cX)=\{Q\subset\cX:\bigcap_{x\in Q}\vor(x,\cX)\neq\emptyset\}=\cA_\infty(\cX).$$
\end{prop}

\begin{figure}
	\begin{subfigure}[b]{0.23\textwidth}
		\centering
		\includegraphics[scale=0.45]{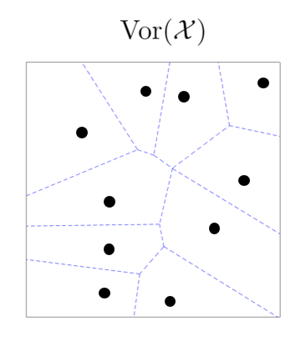}
		\caption{}
	\end{subfigure}
	\begin{subfigure}[b]{0.24\textwidth}
		\centering
		\includegraphics[scale=0.45]{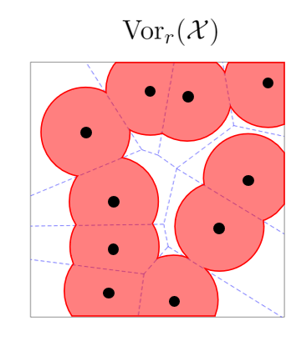}
		\caption{}
	\end{subfigure}
	\begin{subfigure}[b]{0.23\textwidth}
		\centering
		\includegraphics[scale=0.45]{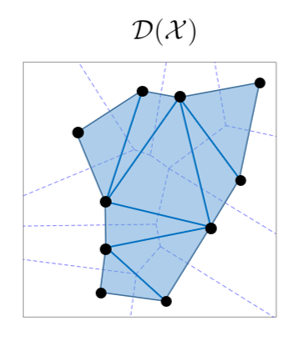}
		\caption{}
	\end{subfigure}
	\begin{subfigure}[b]{0.23\textwidth}
		\centering
		\includegraphics[scale=0.45]{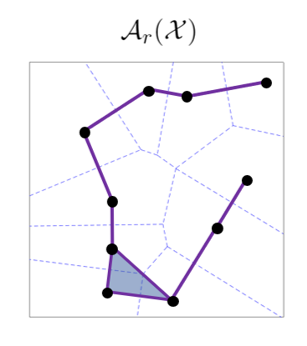}
		\caption{}
	\end{subfigure}
	\caption{Alpha complex. (a) The Voronoi tessellation for a set $\cX\subset\mathbb{R}^2$ in general position. (b) The set of Voronoi balls of radius $r$ of $\cX$. (c) The Delaunay triangulation of $\cX$ which is the dual graph of the Voronoi tessellation. (d) The alpha complex of $\cX$ with radius $r$, which is a subset of the complex generated by taking all the faces of the Delaunay triangulation.}
	\label{fig:alpha_comp}
\end{figure}

Given $\cX\subset\Rd$, we first compute its Delaunay triangulation, which by Proposition \ref{prop: Del_eq_alpha} equals to $\cA_\infty(\cX)$. 
Next, the identification of the subset of simplexes that are in $\cA_r(\cX)$, is done in a top down fashion, starting from the top dimensional simplexes and going downwards.

Let $Q\in\cA_\infty(\cX)$ and denote the filtration value of $Q$ by 
\begin{equation}
    r(Q)=\inf_{r\in\mathbb{R}}\{Q\in\cA_r(\cX)\}=\inf\left\{r\in\mathbb{R} : \bigcap_{x\in Q}\vor_r(x,\cX)\neq\emptyset\right\}.
\end{equation}
In other words, $r(Q)$ is the minimal value for which the Voronoi balls of $Q$ intersect.
The problem of finding $r(Q)$ can be translated into the problem of finding the radius of the minimal $(d-1)$-sphere that includes the points of $Q$ and does not contain any points of $\cX\setminus Q$ in its interior. Thus we obtain the following optimization problem, 
\begin{equation}
    \label{eq: alpha_filt}
    r(Q)=\min_{c\in \bigcap_{x\in Q}\vor(x,\cX)}\|q-c\|\,
\end{equation}
for arbitrary $q\in Q$.
Computing \eqref{eq: alpha_filt} directly is hard. Instead, $r(Q)$ can be computed in a top-down fashion. Define $U(Q)$ to be the set of all $Q$ co-faces $P\in\cA_\infty(\cX)$ of co-dimension $1$.
In \cite{boissonnat2018geometric} the authors show that $r(Q)$ equals to one of two possible values: if the minimal circumsphere of $Q$ does not contain any point of $\cX$ in its interior, then $r(Q)$ equals to the radius of that sphere. Otherwise, $r(Q)=\min_{P\in U(Q)}r(P)$. 
For more information and an explicit algorithm for the computation of the alpha complex see \cite{boissonnat2018geometric}.

\subsection{Persistent homology}
Persistent homology is one of the fundamental tools used in TDA, and can be thought of as a multi-scale version of homology.
While homology is calculated for a single space $X$, persistent homology is applied to a \emph{filtration}.
Let $X$ be a topological space and consider a filtration $\{X_t\}_{t\in \mathbb{R}}$, so that for all $s\le t$ we have $X_s\subset X_t\subset X$. As $t$ is increased, holes can be created and/or filled in, introducing changes to the homology. Persistent homology is used to track these changes.

For $s\le t$, the inclusion $X_s \hookrightarrow X_t$ induces a homomorphism $H_k(X_s)\to H_k(X_t)$ between the homology groups. 
These induced maps enable us to track the evolution of homology classes throughout the filtration, from the point when they are first formed (\emph{born}) to the point when they become boundaries, and hence trivial (\emph{die}).
The algebraic structure tracking this evolution is called a \emph{persistence module}, denoted $\mathrm{PH}_k(X)$. In \cite{zomorodian_computing_2005} it was shown that $\mathrm{PH}_k(X)$ has a unique decomposition into basis elements called \emph{persistence intervals}. Intuitively, each persistence interval tracks a single $k$-cycle from birth to death.
For each persistence cycle $\gamma\in\PH_k(X)$ we denote by 
$\bth(\gamma)$ the point (value of $t$) where $\gamma$ is first created, and by $\dth(\gamma)$  the point where $\gamma$ becomes trivial. The entire lifetime interval is denoted by $\Int(\gamma)=[\bth(\gamma),\dth(\gamma))$. 
In most TDA applications, once persistent homology is calculated one outputs a numerical summary in the form of a \emph{barcode} or a \emph{persistence diagram} (see Figure \ref{fig:Pers_hom_exmp}). These are two equivalent ways to visually represent the collection of $(birth, death)$ pairs for all persistence intervals.

\begin{figure}
    \centering
    \includegraphics[scale=0.7]{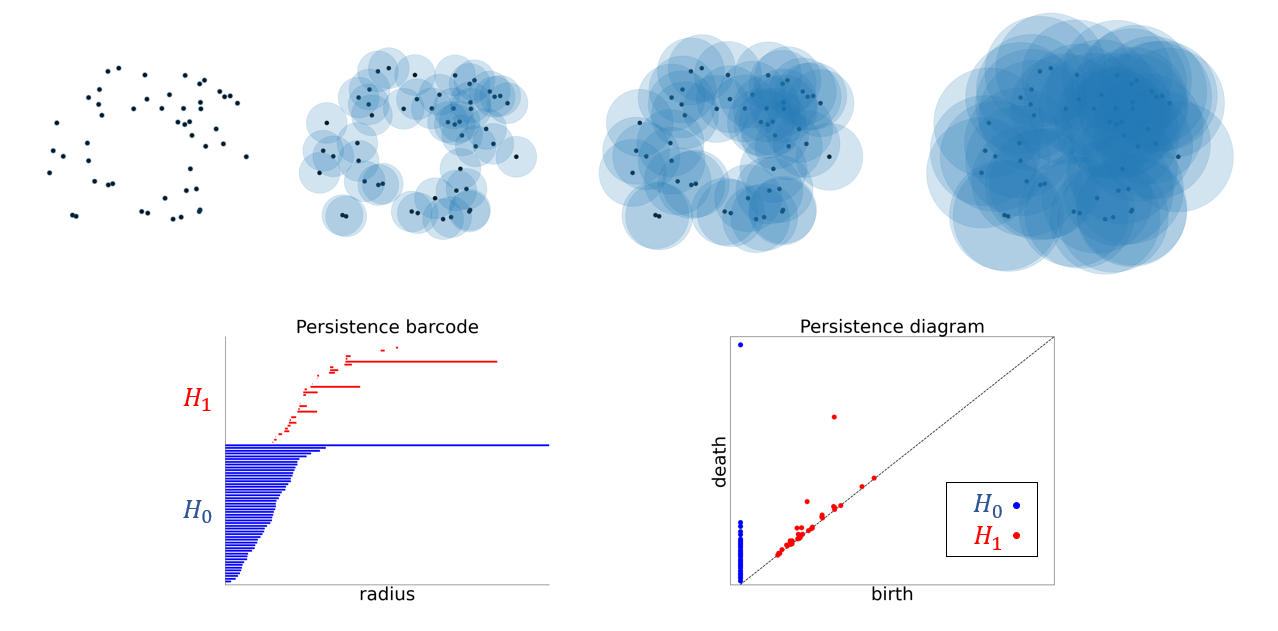}
    \caption{Persistent homology -- example. Top: a continuous filtration induced by inflating the ball cover around a point cloud. Bottom left: The resulting barcode summary for $\PH_0$ and $\PH_1$. Bottom right: The resulting persistence diagram. This example highlights the intuition that long bars (or points far from diagonal) stand for the topological features of the space underlying the point cloud. }
    \label{fig:Pers_hom_exmp}
\end{figure}

An example where persistent homology is used is in the context of geometric complexes. Here, the filtration parameter is the radius $r$, and the persistent homology provides a summary for all the cycles that appear at different scales. Given point cloud data, we can compute the persistent homology of either the \cech or the Rips filtration in order to extract information about the topological space underlying the data.

\section{The Coupled Alpha Complex}
\label{sec:coup_alpha}

In this section, we introduce the \emph{coupled alpha complex} denoted $\cA^\mathrm{co}_r(\cX,\cY)$, where $\cX,\cY\subset\Rd$ are finite sets.
We will define this complex in such a way that it meets the following requirements,
$$\cA_r(\cX)\subset \cA^\mathrm{co}_r(\cX,\cY) \supset\cA_r(\cY),\quad\text{and}\quad\cA^\mathrm{co}_r(\cX,\cY)\simeq\cA_r(\cX\cup \cY).$$
In other words, this complex includes both alpha complexes of each of the sets separately, and is homotopy equivalent to the alpha (or \v{C}ech) complex over their union. In addition, we will show that it maintains low computational costs compared to the \cech and Rips complexes over the union of points.

The \emph{coupled alpha complex} is constructed using the same building blocks as the alpha complex. 
Recall that given a subset $\cX\subset\Rd$, the alpha complex $\cA_r(\cX)$ is defined as the nerve of the Voronoi cells $\vor(x,\cX)$.
The coupled alpha complex is defined using two different Voronoi tessellations, each related to a different set of points, denoted $\cX$ and $\cY$ (see Figure \ref{fig:coupled_alpha_def}). 
The formal definition of this complex is the following.

\begin{definition} [Coupled Alpha Complex]\label{defn:coupled}
Let $\cX,\cY\subset\mathbb{R}^d$ be a pair of finite subsets and let $r\geq 0$. 
Define the following sets
$$\mathcal{V}_r^\cX = \{\vor_r(x,\cX)\}_{x\in \cX},\ \mathcal{V}_r^\cY = \{\vor_r(y,\cY)\}_{y\in \cY},$$
$$\mathcal{V}_r=\mathcal{V}_r^\cX\cup\mathcal{V}_r^\cY.$$
The coupled alpha complex generated by $\cX$ and $\cY$ with parameter $r$ is defined as
$$\cA^\mathrm{co}_r(\cX,\cY) := \cN(\mathcal{V}_r).$$
\end{definition}
Note that the (standard) alpha complex on $\cX\cup\cY$ is given by taking the nerve of $\mathcal{V}_r^{\cX\cup\cY}$.

\begin{figure}[h]
    \centering
    \includegraphics[scale=0.6]{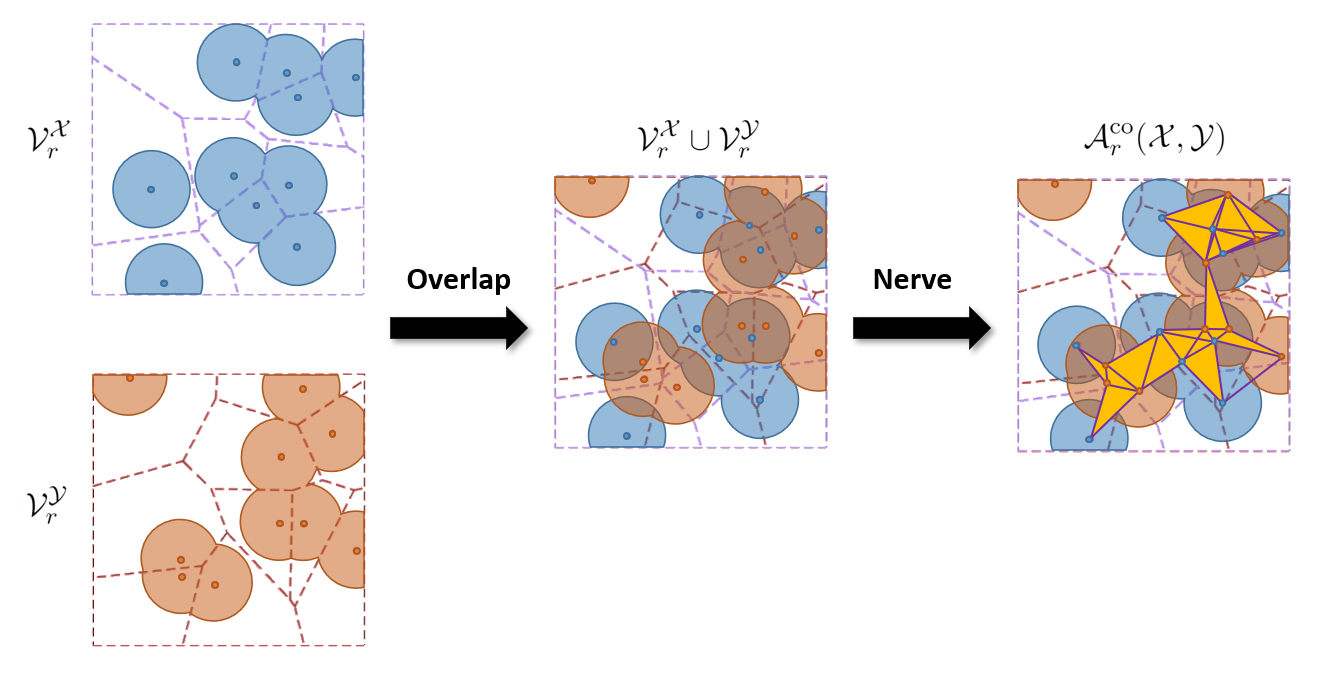}
    \caption{Coupled alpha complex definition. $\cX,\cY\subset\mathbb{R}^2$ and $\mathcal{V}_r^{\cX},\mathcal{V}_r^{\cY}$ are their corresponding Voronoi balls of radius $r$ (left). We ``overlap'' the sets (middle) and then take the nerve in order to get $\cA^\mathrm{co}_r(\cX,\cY)$ (right).}
    \label{fig:coupled_alpha_def}
\end{figure}

\begin{remark}
For points in general position in $\R^d$, the dimension of the alpha complex is always $d$. However, it is important to notice that the coupled alpha complex is $(d+1)$-dimensional. This stems, for example, from configurations where the intersection point of $d$ Voronoi cells in $\cX$ lies in the interior of a Voronoi cell in $\cY$. Nonetheless, by the nerve lemma we have $H_d,H_{d+1}=0$. 
\end{remark}

The definition above implies that the following inclusion relations hold,
$$\cA_r(\cX)\hookrightarrow \cA^\mathrm{co}_r(\cX,\cY) \hookleftarrow\cA_r(\cY).$$
In addition, by the Nerve Lemma \ref{lemma:nerve} we have the following.

\begin{lem}
Let $\cX,\cY$ be a pair of finite subsets of $\mathbb{R}^d$ and let $r\geq 0$. Then,
$$\cA^\mathrm{co}_r(\cX,\cY)\simeq \cA_r(\cX\cup\cY).$$
\end{lem}

\begin{proof}
The elements of $\mathcal{V}_r$ are all convex sets (as intersections of convex sets). 
Hence, by the Nerve Lemma \ref{lemma:nerve}
$$\cN(\mathcal{V}_r)\simeq \bigcup_{x\in \cX} \vor_r(x,\cX)\bigcup_{y\in \cY} \vor_r(y,\cY)=\bigcup_{z\in \cX\cup \cY}B_r(z),$$
and $\bigcup_{z\in \cX\cup \cY}B_r(z) \simeq \cA_r(\cX\cup\cY)$.
\end{proof}

In conclusion, we have the following relations, where the new complex $\cA^\mathrm{co}_r(\cX,\cY)$ substitutes $\cA_r(\cX\cup \cY)$. The dashed arrows represent the  missing inclusion relations between $\cA_r(\cX),\cA_r(\cY)$ and $\cA_r(\cX\cup \cY)$ that are replaced by the inclusion maps into $\cA^\mathrm{co}_r(\cX,\cY)$.
\[
\centering
\begin{tikzcd}
\cC_r(\cX) \arrow[r,hook,"i"] & \cC_r(\cX\cup \cY) &\arrow[l, hook',  "i" above]  \cC_r(\cY)\\
\cA_r(\cX) \arrow[dr, dashed] \arrow[u,hook,"\simeq"] \arrow[r,hook,"i"] 
	& \cA^\mathrm{co}_r(\cX, \cY)  \arrow[u,hook, "\simeq"] & 
	    \arrow[dl,dashed] \cA_r(\cY) \arrow[u,hook,"\simeq"] \arrow[l, hook', "i" above]\\
	& \cA_r(\cX\cup \cY) \arrow[u,hook, "\simeq" left] & 
\end{tikzcd}
\]

\section{Constructing the Coupled Alpha Complex} 
\label{sec:construct}
In the following section we provide an algorithm for calculating the coupled alpha filtration $\{\cA^\mathrm{co}_r(\cX,\cY)\}_{r\geq 0}$. The computation is divided into two steps. We start by finding $\cA^\mathrm{co}_{\infty}(\cX,\cY)$, i.e.~the set of all possible simplexes that may appear in the coupled alpha filtration. Then, we calculate the filtration values for each possible simplex.

\subsection{Computing $\coA_{\infty}(\cX,\cY)$}
Assume that $\cX$ and $\cY$ are assigned with a given ordering.
The key idea is to lift these sets from $\R^d$ to $\Rdd$ in such a way that the lifted sets $\hat\cX$ and $\hat\cY$ lie in parallel hyperplanes.
We proceed by computing the regular $(d+1)$-dimensional Delaunay triangulation of $\hat\cX\cup\hat\cY$ and its induced simplicial complex denoted $\cD(\hat\cX\cup \hat\cY)$.  Finally, we show that $\cD(\hat\cX\cup \hat\cY)$ is isomorphic to $\cA^\mathrm{co}_\infty(\cX,\cY)$.

Let $\pi:\mathbb{R}^{d+1}\to\mathbb{R}^d$ be the natural projection, i.e.~$(x_1,...,x_d,x_{d+1})\overset{\pi}{\mapsto} (x_1,...,x_d)$. Let $Q=[v_0,...,v_n]$ be a simplex such that $\{v_0,...,v_n\}\subset\mathbb{R}^{d+1}$. 
The projected simplex of $Q$, denoted by $\pi(Q)$, is defined by:
\begin{equation}\label{eqn:pi_Q}
\pi(Q)=\pi([v_0,...,v_n])=[\pi(v_0),...,\pi(v_n)].
\end{equation}
Let $(\cX,\cY)$ be an ordered pair of finite subsets of $\mathbb{R}^d$. Define $\hat{\mathcal{X}}\triangleq \cX \times\{0\}$ and $\hat{\mathcal{Y}}\triangleq \cY\times\{1\}$.
\begin{figure}
	\centering 	
	\includegraphics[scale=0.6]{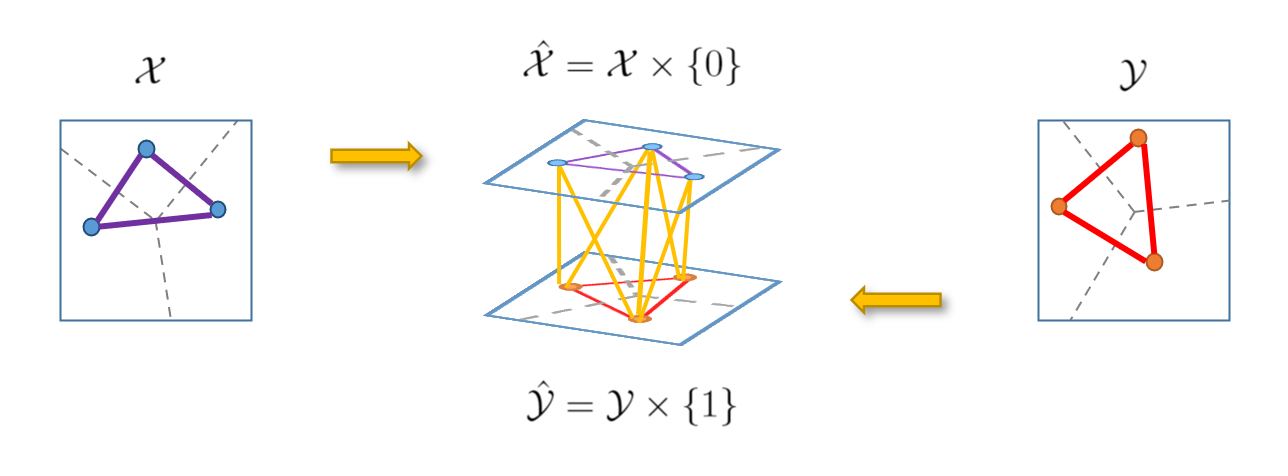}
	\caption{Computing $\cA^\mathrm{co}_{\infty}(\cX,\cY)$. Two sets of points in $\mathbb{R}^2$ are embedded in two parallel hyperplanes in $\mathbb{R}^3$. The triangulation in $\mathbb{R}^3$ restricted to one of the sets equals to its triangulation in $\mathbb{R}^2$.
    }
	\label{fig:A_inf_comp}
\end{figure}
Throughout, we will assume that the sets $\cX,\cY\subset \R^d$ are in \emph{coupled general position}, defined as follows.
\begin{definition}[Coupled General Position] \label{def:coup_gen_pos}
Let $P_1,P_2\subset\Rd$ be two finite sets and define $\hat{P_1}=P_1\times\{0\}$ and $\hat{P_2}=P_2\times\{1\}$.
$P_1$ and $P_2$ are said to be in \emph{coupled general position} if: \begin{enumerate}
    \item $P_1$ and $P_2$ are in general position (in $\Rd$), and 
    \item for every $\hat{Q}_1\subset\hat{P}_1$ and $\hat{Q}_2\subset\hat{P}_2$ such that $|\hat{Q}_1\cup\hat{Q}_2|=d+2$, no point of $(\hat{P}_1\setminus\hat{Q}_1)\cup(\hat{P}_2\setminus \hat{Q}_2)$ lies on the circumsphere of $\hat{Q}_1\cup\hat{Q}_2$.
\end{enumerate}   
\end{definition}
This assumption implies that the vertices of the Voronoi tessellation of $\hat\cX\cup\hat\cY$ are at the intersection of exactly $d+2$ Voronoi cells (see the proof of Proposition \ref{prop:1} below).
Note that for generic \emph{random} point processes, this assumption holds with probability $1$. This is true, since for any $Q\subset\hat\cX\cup\hat\cY$ of size $d+2$ that contains points in both sets $\hat\cX,\hat\cY$, the probability that a point $p\in\hat\cX\cup\hat\cY\setminus Q$ lies on the circumsphere of $Q$ is zero (since the circumsphere is a set of $0$-measure).

In the following, we argue that the coupled general position assumption is sufficient for the Delaunay triangulation $\cD(\hat\cX\cup \hat\cY)$ to be uniquely-defined.
Recall that for points in general position the Delaunay triangulation is unique. From Definition \ref{def:gen_pos}, taking $P=\hat{\cX}\cup\hat{\cY}$, the conditions hold for sets $Q$ that are composed of points from both sets $\hat\cX$ and $\hat\cY$ by the coupled general position assumption. 
However, general position is violated,  
for $Q\subset\hat\cX$ (or $Q\subset\hat\cY$) of size $(d+2)$, as such sets lie on a $d$-dimensional flat (the containing hyperplane). 
However, these sets are ignored to get a valid triangulation - a subdivision of the convex hull of $\hat\cX\cup\hat\cY$ into $(d+1)$-simplexes that form a simplicial complex - which is unique and dual to the Voronoi diagram of $\hat\cX\cup\hat\cY$.

\begin{algorithm}
	\caption{Computing $\cA^\mathrm{co}_{\infty}(\cX,\cY)$}
	\label{alg: 1}
	\SetKwInOut{Input}{input}\SetKwInOut{Output}{output} 
	\Input{ $\cX$ and $\cY$. \\}
	\Output{ $\pi(\cD)=\cA^\mathrm{co}_\infty(\cX,\cY)$.}
	\BlankLine
	$\hat\cX\leftarrow \cX\times\{0\}$\;
	$\hat\cY\leftarrow \cY\times\{1\}$\;
	$\cD\leftarrow$ \textbf{delaunayTriangulation}($\hat\cX\cup\hat\cY$)\;
	\textbf{return} $\pi(\cD)$\;
\end{algorithm}

\begin{prop} 
\label{prop:1}
Let $(\cX,\cY)$ be an ordered pair of finite subsets of $\mathbb{R}^d$. 
Then,
$$\cA^\mathrm{co}_\infty(\cX,\cY) = \pi(\cD(\hat\cX\cup\hat\cY)),$$
where the right-hand-side is the abstract simplicial complex generated by projecting the coordinates of the faces of $\cD(\hat\cX\cup\hat\cY)$ using $\pi(\cdot)$ \eqref{eqn:pi_Q}.
\end{prop}

\begin{proof}
First we show that if $Q\in\cA^\mathrm{co}_\infty(\cX,\cY)$ then $Q\in\pi(\cD(\hat\cX\cup\hat\cY))$. 
Denote by $Q_{\cX}=Q\cap \cX$ and by $Q_{\cY}=Q\cap {\cY}$, and assume that both $Q_{\cX}$ and $Q_{\cY}$ are not empty. 
Since $Q\in\cA^\mathrm{co}_{\infty}(\cX,\cY)$
$$I=\bigcap_{x\in Q_\cX}\vor(x,\cX)\bigcap_{y\in Q_\cY} \vor(y,\cY)\neq\emptyset.$$
Let $h_0,h_1:\Rd\rightarrow\Rdd$ be such that $h_0(z)=(z,0)$ and $h_1(z) = (z,1)$.
In order to show that $Q\in\pi(\cD(\hat\cX\cup\hat\cY))$, we need to show that the following condition holds.
$$\hat I=\bigcap_{x\in Q_\cX} \vor\big(h_0(x),\hat\cX\big)\bigcap_{y\in Q_\cY} \vor\big(h_1(y),\hat\cY\big)\neq\emptyset.$$
Define $s=(z,t)\in\Rdd$ such that $z\in I$ and $t=\dfrac{1}{2}(\|y-z\|^2-\|x-z\|^2+1)$ for arbitrary $x\in Q_\cX$ and $y\in Q_\cY$.
For this choice of $s$ we have that
for all $x'\in Q_\cX$ and $y'\in Q_\cY$,
$$\|h_0(x')-s\|^2=\|h_1(y')-s\|^2,$$
and for every $x'\in Q_\cX$ and $y''\in \cY\setminus Q_\cY$
$$\|h_0(x')-s\|^2=\|h_1(y')-s\|^2=\|y'-z\|^2+(1-t)^2\leq \|y''-z\|^2+(1-t)^2=\|h_1(y'')-s\|^2$$
where $y'\in Q_\cY$ is an arbitrary point.
Similarly, for every $y'\in Q_\cY$ and $\forall x''\in \cX\setminus Q_\cX$
$$\|h_1(y')-s\|^2\leq \|h_0(x'')-s\|^2.$$
Hence, $s\in\hat I$ which implies that $\hat I$ is non-empty and therefore $Q\in\pi(\cD(\hat\cX\cup\hat\cY))$.

So far we assumed that $Q_\cX,Q_\cY$ are non-empty. Next, assume without loss of generality, that $Q_\cY$ is empty. In that case 
$$I=\bigcap_{x\in Q}\vor(x,\cX)\text{ and } \hat I=\bigcap_{x\in Q}\vor(h_0(x),\hat\cX).$$
Define $s=(z,t^*)$ such that $z\in I$ and $t^*=\min_{y\in \cY}\dfrac{1}{2}(\|y-z\|^2-\|x-z\|^2+1)$ for an arbitrary point $x\in Q$.
For this choice of $s$ we get that for all $x'\in Q$ and $y\in \cY$
$$t^*\leq \frac{1}{2}(\|y-z\|^2-\|x'-z\|^2+1)$$
implying that,
$$ \|x'-z\|^2 +(t^*)^2\leq \|y-z\|^2 +1 - 2t^* +(t^*)^2$$
which leads to
$$ \|h_0(x')-s\|^2\leq \|h_1(y)-s\|^2.$$
Hence, $s\in \hat I$ which completes the proof of the first direction.

Next, we need to show that if $Q\in \pi(\cD(\hat\cX\cup\hat\cY))$ then $Q\in \cA^\mathrm{co}_\infty(\cX,\cY)$.
Denote by $\hat Q=h_0(Q_\cX)\cup h_1(Q_\cY)$, i.e., the simplex that is obtained by applying $h_0$ to the points of $Q_\cX$ and $h_1$ to the points of $Q_\cY$.
By assumption, there exists $s=(s_1,...,s_{d+1})\in \Rdd$ such that $\forall q,q'\in \hat Q,\ and\ w\in \hat \cX\cup\hat\cY$
$$\|q-s\|^2=\|q'-s\|^2,\quad\|q-s\|^2\leq\|w-s\|^2.$$
Denote $z=\pi(s)$. If $q\in \hat Q\cap\hat\cX$, then $\forall w\in\hat\cX$
$$\|\pi(q)-z\|^2 = \|q-s\|^2 - s_{d+1}^2 \leq \|w-s\|^2 - s_{d+1}^2= \|\pi(w)-z\|^2,$$
hence,
$$z\in\bigcap_{q\in\hat \cX\cap \hat Q}\vor(\pi(q),\cX)=\bigcap_{x\in Q_\cX}\vor(x,\cX).$$
Similarly, we can show
$$z\in\bigcap_{q\in\hat\cY\cap \hat Q}\vor(\pi(q),\cY)=\bigcap_{y\in Q_\cY}\vor(y,\cY)$$
Thus, $z\in \bigcap_{x\in Q_\cX}\vor(x,\cX)\bigcap_{y\in Q_\cY} \vor(y,\cY)$, which completes the proof.
\end{proof}

To conclude, in this section we showed that Algorithm \ref{alg: 1} produces the complex $\cA^\mathrm{co}_{\infty}(\cX,\cY)$.

\subsection{Computing filtration values }
\label{sec: filt_vals}
Once we computed $\cA^\mathrm{co}_\infty(\cX,\cY)$,
we proceed to determining the filtration value attached to each simplex. 
Recall that in the alpha complex, the filtration value attached to each simplex is given by the radius of the minimal circumsphere that does not contain any other points in its interior \eqref{eq: alpha_filt}. We will define the filtration value of a simplex $Q\in \cA^\mathrm{co}_\infty(\cX,\cY)$ in a similar fashion. Let $Q_\cX = Q\cap \cX$ and $Q_\cY=Q\cap \cY$. We define the filtration value of $Q$ to be:
\begin{equation}
\label{eq: r_Q}
    r(Q) = \inf\Big\{r\geq 0 : \bigcap_{q\in Q_\cX}\Big(\vor(q,\cX)\cap B_r(q)\Big)\bigcap_{q\in Q_\cY}\Big(\vor(q,\cY)\cap B_r(q)\Big)\neq\emptyset \Big\}.
\end{equation}
The optimization problem stated in \eqref{eq: r_Q} can be translated to a problem where the \emph{unknown variable} is the center (in $\Rd$) of two \emph{concentric} $d$-dimensional open balls $B_\cX$ and $B_\cY$ that are empty in the following sense.
The first ball $B_\cX$ includes the points $Q_\cX$ on its boundary and it does not include any other point of $\cX$ in its interior. Similarly, the second ball $B_\cY$ includes the points $Q_\cY$ on its boundary and it does not include any other point of $\cY$ in its interior.
In general, there can be an infinite number of such pairs of balls (see Figure \ref
{fig:spheres_pairs}).

For each such pair we can define $r_{\max}$ to be the radius of the bigger ball. Then, we set $r(Q)=\min r_{\max}$, where the minimum is taken over all possible pairs  $(B_{\cX},B_\cY)$. 
In the following we state the optimization problem whose solution is the filtration value of a given simplex, and present algorithm for computing the minimizer of this problem. 

\begin{figure}
	\centering 	
	\includegraphics[scale=0.35]{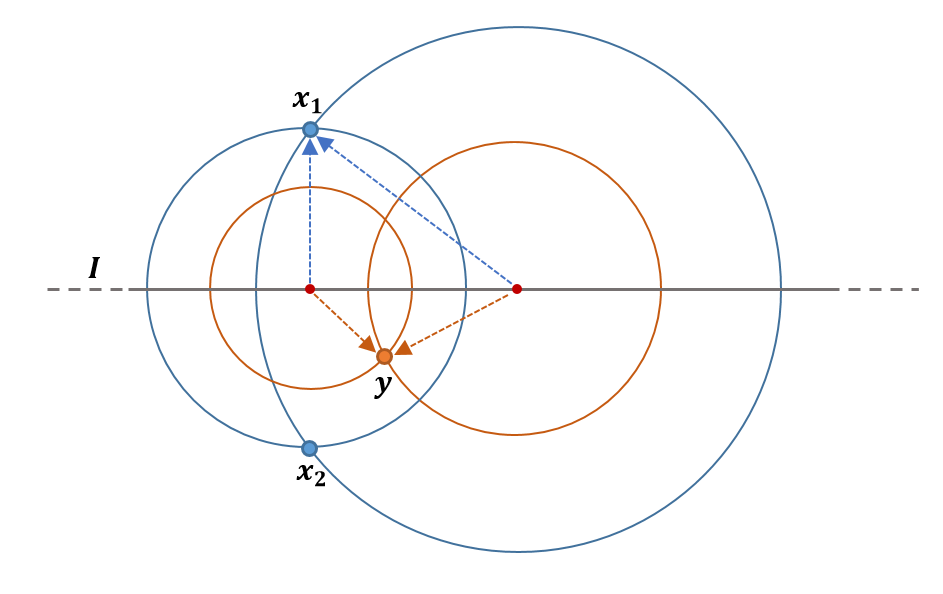}
	\caption{Given two sets $\cX=\{x_1,x_2\}$ and $\cY=\{y\}$, there is an infinite number of pairs of concentric spheres with center in $I=\vor(x_1,\cX)\cap\vor(x_1,\cX)\cap\vor(y,\cY)$, such that one includes $\cX$ and the other one includes $\cY$. However, there is only one pair with minimal $r_{\max}$ (the left pair in the image).
	}
	\label{fig:spheres_pairs}
\end{figure}

\subsubsection{Optimization problem}
\label{sec:opt_prob}
Let $Q\in \cA^\mathrm{co}_\infty(\cX,\cY)$, such that both $Q_\cX=\cX\cap Q$ and $Q_\cY=\cY\cap Q$ are not empty and assume that $|Q|=m+1$ and $|Q_\cX| = l$. 
Let 
\begin{equation}\label{eq:IXY}
I(Q,\cX):=\bigcap_{x\in Q_\cX}\vor(x, \cX),\quad\text{and}\quad I(Q,\cY):=\bigcap_{x\in Q_\cY}\vor(y, \cY).
\end{equation}
The filtration value of $Q$ is the solution for the following optimization problem:
\begin{equation}
\label{eq:opt1}
r^2(Q)=\min_{c\in I(Q,\cX)\cap I(Q,\cY)}\max\{\|c-x\|^2,\|c-y\|^2\},
\end{equation}
where $x\in Q_\cX$ and $y\in Q_\cY$ are two arbitrary points.

Note that in the case where $Q_\cY$ is empty, i.e. $Q_\cX=Q$, the solution for \eqref{eq:opt1} coincides with the filtration value of $Q$ as a simplex in the alpha complex of $\cX$ and vice versa.

The minimizer of (\ref{eq:opt1}) must lie in the intersection of the Voronoi cells of the vertices of $Q$.
This constraint makes the optimization problem hard to solve.
Instead, we will first solve the following relaxed version of the problem
\begin{equation}\label{eq:opt1_relaxed}
r_{\mathrm{rel}}^2(Q)=\min_{c\in I(Q,Q_\cX)\cap I(Q,Q_\cY)}\max\{\|c-x\|^2,\|c-y\|^2\},
\end{equation}
where we replaced $\cX,\cY$ with $Q_\cX, Q_\cY$, respectively. Note that $r_\mathrm{rel}(Q)\leq r(Q)$ since $I(Q,\cX)\cap I(Q,\cY)\subset I(Q,Q_\cX)\cap I(Q,Q_\cY)$.
The solution for the relaxed problem will play a central role in identifying the filtration value of $Q$, i.e.~the solution for (\ref{eq:opt1}), as we will see in the following.

\subsubsection{Relaxed optimization problem}
In the following, we reformulate the optimization problem \eqref{eq:opt1_relaxed}.
Assume an arbitrary ordering on $Q_\cX=\{x_1,...,x_l\}$ and $Q_\cY=\{y_1,...,y_{m-l+1}\}$. The set
$I(Q,Q_\cX)$ consists of all points $c\in\mathbb{R}^d$ that satisfy the following $l-1$ equations,
$$\|c-x_1\|^2=\|c-x_{i}\|^2,\quad\forall i\in\{2,...,l\}.$$
Similarly, $I(Q,Q_\cY)$ is the set of all points $c\in\mathbb{R}^d$ that  satisfy
$$\|c-y_1\|^2=\|c-y_i\|^2,\quad\forall i\in\{2,...,m-l+1 \}.$$
The solution we seek lies in $I(Q,Q_\cX)\cap I(Q,Q_\cY)$, hence, we combine the equations above into a single system. For simplicity, we write the equations in a matrix form.
Define,\[
A
=
\begin{bmatrix}
(x_2-x_1)^T \\
\vdots \\ 
(x_l-x_1)^T\\
(y_2-y_1)^T \\
\vdots \\
(y_{m-l+1}-y_1)^T
\end{bmatrix}
,\ 
c=
\begin{bmatrix}
c_1 \\
\vdots \\
c_{d}
\end{bmatrix}
,\ 
b
=
\frac{1}{2}
\begin{bmatrix}
\|x_2\|^2-\|x_1\|^2 \\
\vdots \\
\|x_l\|^2-\|x_1\|^2  \\
\|y_2\|^2-\|y_1\|^2 \\
\vdots \\
\|y_{m-l+1}\|^2-\|y_1\|^2 
\end{bmatrix}
.
\]
Note that since the points are in general position, the rows of $A$ are linearly independent, i.e.~$\mathrm{rank}(A)=m-1$.
Using the above notation, the optimization problem can now be written in the following way
\begin{equation}
\label{eq:opt2}
 r_{\mathrm{rel}}^2(Q) = \min_{c\in\Rd} \max\{\|c-x_1\|^2,\|c-y_1\|^2\},\quad\text{subject to }Ac=b.
\end{equation}
Note that for all $i,j$ we have $\|c-x_i\| = \|c-x_1\|$ and $\|c-y_j\| = \|c-y_1\|$, and the choice of $x_1,y_1$ in \eqref{eq:opt2} is arbitrary.
Using the fact that a solution for a linear system can be expressed as a sum of a solution for the homogeneous system and a particular solution, we reformulate the optimization problem in the following way,
\begin{equation}
\label{eq:opt3}
r_{\mathrm{rel}}^2(Q) = \min_{s\in\mathbb{R}^{d-m+1}} \max\{\|Fs+c_0-x_1\|^2,\|Fs+c_0-y_1\|^2\}, 
\end{equation}
where $c_0$ is a particular solution, i.e. $Ac_0=b$, and $F\in \mathbb{R}^{d\times d-m+1}$ is the matrix whose columns are an orthonormal basis of $\Ker(A)$.

We derive the solution in the following way.
Denote by $r^2(Q,s)$ the objective function in \eqref{eq:opt3}, then its derivative  with respect to $s$ is
\[ \frac{dr^2(Q,s)}{ds} =
	\begin{cases}
    		F^T(Fs+c_0-x_1)       & \quad \text{if } \|Fs+c_0-x_1\|^2 \geq \|Fs+c_0-y_1\|^2,\\
    F^T(Fs+c_0-y_1)  & \quad \text{if } \|Fs+c_0-x_1\|^2 < \|Fs+c_0-y_1\|^2.
  \end{cases}
\]
Recalling that $F$ consists of orthonormal columns, the minimizer of (\ref{eq:opt3}) is one of the following three terms,
$$s_{x}=F^T(x_1-c_0),$$
or
$$s_{y}=F^T(y_1-c_0),$$
or 
$$s_0 = F^{T}(C(Q)-c_0),$$
where $C(Q)$ is the center of the minimal circumsphere of $Q$ (in $\mathbb{R}^d$).

In practice, we compute the first two candidates $s_x,s_y$, measure their distances to $y_1$ and $x_1$ and then check which is the right solution according to the following claim. 
\begin{lem}
Denote
$$r_x(x_1) = \|Fs_x+c_0-x_1\|,$$
$$r_x(y_1) = \|Fs_x+c_0-y_1\|,$$
$$r_y(x_1) = \|Fs_y+c_0-x_1\|,$$
$$r_y(y_1) = \|Fs_y+c_0-y_1\|.$$
The solution for (\ref{eq:opt2}), denoted by $r_{\mathrm{rel}}(Q)$, is given by 
\[r_{\mathrm{rel}}(Q) = 
	\begin{cases}
    		r_x(x_1)       & \quad \text{if } r_x(x_1) \geq r_x(y_1),\\
    		r_y(y_1)       & \quad \text{if } r_y(x_1) \leq r_y(y_1),\\
    		R(Q)      & \quad \text{otherwise},\\
  \end{cases}
\]
where $R(Q)$ is the radius of the minimal circumsphere of $Q$ (in $\mathbb{R}^d$).
\end{lem}

\begin{proof} The proof relies on  the fact that the objective in \eqref{eq:opt2} is convex.  
There are three cases we should address. The first case is when $r_x(x_1)\geq r_x(y_1)$. 
Here, since $s_x$ is the global minimum of $\|Fs+c_0-x_1\|^2$ in $\R^{d-m+1}$, we have that $r_x(x_1) \le r_y(x_1)$, and similarly $r_y(y_1)\le r_x(y_1)$.
Therefore, we conclude that $r_y(x_1)\geq r_y(y_1)$, which implies that $r_{\mathrm{rel}}(Q)=r_x(x_1)$.  The second case, when $r_y(y_1)\geq r_y(x_1)$ is similar.
The last case is when both $r_x(x_1)<r_x(y_1)$ and $r_y(y_1)<r_y(x_1)$. Here, the minimum of (\ref{eq:opt3}) is a boundary point, i.e.~a value $s\in\mathbb{R}^{d-m+1}$ that satisfies
$$\|Fs+c_0-x_1\|^2=\|Fs+c_0-y_1\|^2$$
(otherwise, the minimizer must be either of $s_x$ and $s_y$, in contradiction to the assumption).
Under this constraint, (\ref{eq:opt3}) is exactly the optimization problem for finding the minimal circumsphere of $Q$. Hence, the minimizer in that case satisfies $C(Q)=Fs_*+c_0$, which gives $r_{\mathrm{rel}}(Q)=R(Q)$.
\end{proof}

\begin{figure}
	\label{fig:opt_sol}
	\centering
	\includegraphics[scale=0.35]{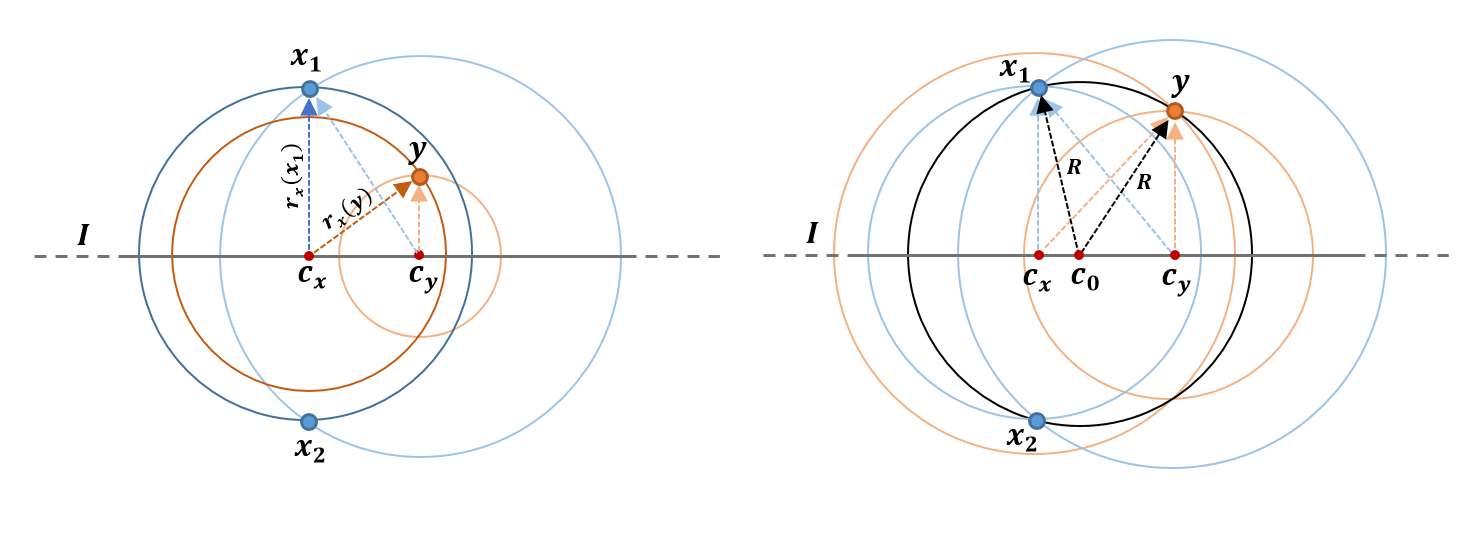}
	\caption{Possible solutions for the optimization problem (\ref{eq:opt2}), in $\mathbb{R}^2$, for the case $Q_\cX=\{x_1,x_2\}$ and $Q_\cY=\{y\}$. $I$ denotes the intersection $I(Q,Q_\cX)\cap I(Q,Q_\cY)$. The centers $c_x,c_y$ and $c_0$ are equal to $Fs_x,Fs_y$ and $Fs_0$ respectively, where $s_x,s_y$ and $s_0$ are the candidates for the minimizer of \eqref{eq:opt3}. The figure on the left depicts a case where $r_{\mathrm{rel}}(Q)=r_x(x_1)$. the figure on the right depicts a case where $r_{\mathrm{rel}}(Q)=R=R(Q)$. 
	}
\end{figure}

\subsubsection{Back to the non-relaxed problem}
Determining the solution for the non-relaxed problem (\ref{eq:opt1}) can be done in a top-down fashion, similarly to the calculation of the filtration values for the alpha complex \cite{boissonnat2018geometric}.
Adopting the terminology of \cite{boissonnat2018geometric},  we define the following condition.
 
Let $P=P_{\cX}\cup P_{\cY}\in \cA^\mathrm{co}_\infty(\cX,\cY)$ and let $Q=Q_{\cX}\cup Q_{\cY}$ be a  face of $P$ (of co-dimension $1$).
Let $c_*$ be the minimizer of \eqref{eq:opt1_relaxed} for $Q$.
Define $B_\cX$ to be the open ball with center $c_*$ and radius $\|c_*-x\|$ for arbitrary $x\in Q_\cX$, and let $S_\cX$ be its bounding sphere. Similarly define $B_\cY$ and $S_\cY$.

We say that the ordered pair $(P,Q)$ satisfies the \emph{coupled Gabriel condition} if 
both $B_\cX\cap P_\cX = \emptyset$, and $B_\cY\cap P_\cY = \emptyset$.
Let $U(Q)\subset\cA^\mathrm{co}_\infty(\cX,\cY)$ be the set of all co-faces of $Q$ (of co-dimension $1$). If for every $P\in U(Q)$, the pair $(P,Q)$ satisfies the coupled Gabriel condition, then we say that $Q$ is a \emph{coupled Gabriel simplex}.
This property will enable us to determine whether the solution of (\ref{eq:opt1_relaxed}) for $Q$ identifies with the solution of (\ref{eq:opt1}),
as we suggest in the following lemma (see Figure \ref{fig:coupled_gab}).

\begin{lem}\label{lem:filt_1}
Let $Q \in \cA^\mathrm{co}_\infty(\cX,\cY)$ be a coupled Gabriel simplex. Then, the filtration value of  $Q$ is given by $r(Q) = r_{\mathrm{rel}}(Q)$.
\end{lem}

\begin{proof}
Our goal is to show that the solutions of  \eqref{eq:opt1} and  \eqref{eq:opt1_relaxed} for $Q$ are the same.
In the previous section we found  the minimizer $c_*$ for the relaxed problem \eqref{eq:opt1_relaxed}.
If we can show that $c_*\in I(Q,\cX)\cap I(Q,\cY)$, then $c_*$ also minimizes \eqref{eq:opt1}, and therefore $r(Q) = r_{\mathrm{rel}}(Q)$.

Let $P\in U(Q)$, and without loss of generality, suppose that $v = P\setminus Q \in \cX$. The coupled Gabriel condition implies that $\|c_*-v\| \ge \|c_*-x\|$ for an arbitrary $x\in Q_\cX$. Since this is true for every $P$ we conclude that $c_*\in \bigcap_{x\in Q_\cX}\vor(x, \cX)$.
Similarly, we can show that $c_*\in \bigcap_{y\in Q_\cY}\vor(y, \cY)$, concluding the proof.
\end{proof}

The following lemma provides a way to determine the filtration value of $Q$ in the case where the conditions of Lemma \ref{lem:filt_1} are not satisfied, i.e.~the coupled Gabriel condition does not hold for $Q$.

\begin{lem}
Let $Q\in\cA^\mathrm{co}_\infty(\cX,\cY)$, and suppose there exists $P\in U(Q)$ such that $(P,Q)$ does not satisfy the coupled Gabriel condition. Then, $$r(Q) = \min_{P\in U(Q)}r(P).$$
\end{lem}

\begin{proof}
For every $P\in U(Q)$, $Q$ is a face of $P$, which implies that $I(P,\cX)\cap I(P,\cY)\subset I(Q,\cX)\cap I(Q,\cY)$, hence  $r(Q)\leq\min_{P\in U(Q)}r(P)$.
Assume that $r(Q)<\min_{P\in U(Q)}r(P)$. In that case, the solution of \eqref{eq:opt1} must lie in the interior of $I(Q,\cX)\cap I(Q,\cY)$. 
Let $c_*$ be the minimizer that corresponds to $r(Q)$. Since $\eqref{eq:opt1}$ is a convex optimization problem and $I(Q,\cX)\cap I(Q,\cY)$ is a convex polyhedron, $c_*$ also solves the relaxed optimization problem \eqref{eq:opt1_relaxed} for $Q$ (since the relaxed constraint is the affine hull of the non-relaxed constraint $I(Q,\cX)\cap I(Q,\cY)$). This is in contradiction to $Q$ not being coupled Gabriel simplex.
\end{proof}

\begin{figure}[h]
	\begin{subfigure}{0.32\textwidth}
		\centering
		\includegraphics[width=\textwidth]{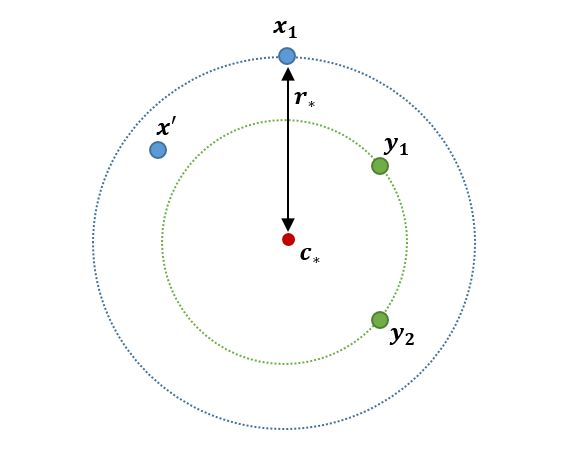}
		\caption{}
	\end{subfigure}
	~
	\begin{subfigure}{0.32\textwidth}
		\centering
		\includegraphics[width=\textwidth]{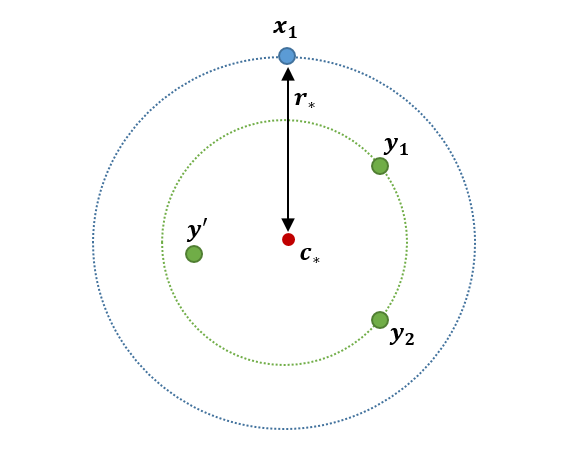}
		\caption{}
	\end{subfigure}
	~
	\begin{subfigure}{0.32\textwidth}
		\centering
		\includegraphics[width=\textwidth]{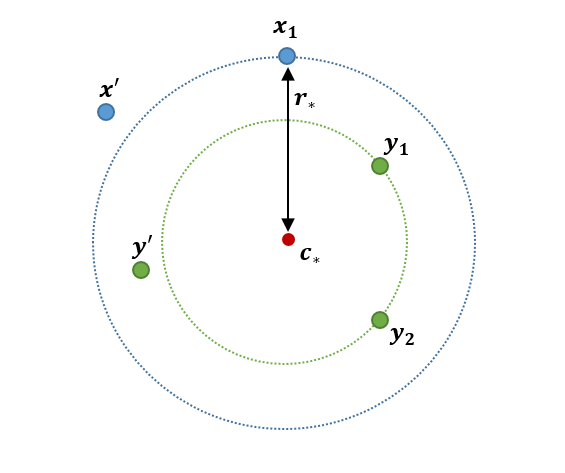}
		\caption{}
	\end{subfigure}
	\caption{Relaxed optimization problem solution validity in $\mathbb{R}^2$, for the case $Q_\cX=\{x_1\}$ and $Q_\cY=\{y_1,y_2\}$. In (a) and (b) the solution is not valid while in (c) it is. The point $c_*$ is the minimizer of \eqref{eq:opt2}. (a) $S_\cX$ (dashed blue line) includes a point $x'\in \cX\setminus Q_\cX$, hence $(P,Q)$ are not coupled Gabriel, where $P=\{x_1,y_1,y_2,x'\}$. (b) $S_\cY$ (dashed green line) includes a point $y'\in \cY\setminus Q_\cY$, hence $(P,Q)$ are not coupled Gabriel, where $P=\{x_1,y_1,y_2,y'\}$. (c) $S_\cX, S_\cY$ excludes all the points of $\cX,\cY$ respectively, hence, the solution $c_*,r_*$ is valid. 
	}
	\label{fig:coupled_gab}
\end{figure}

To conclude, in this section we introduced a way to determine the filtration values for the simplexes in $\cA^\mathrm{co}_{\infty}(\cX,\cY)$. Algorithm \ref{alg:coupled_alpha_filt} produces the filtration values for all the simplexes in $\cA^\mathrm{co}_{\infty}(\cX,\cY)$ sequentially, starting from top dimensional simplexes and moving downwards. The function \textbf{filtrationValue}, first appearing in line 5, returns the minimizer of the relaxed problem \eqref{eq:opt3} for given simplex. Note that in practice, we index the vertices of the sets in the following way $\cX=\{v_1,...,v_{|\cX|}\}$ and $\cY=\{v_{|\cX|+1},...,v_{|\cX|+|\cY|}\}$, so that given a simplex $P$ it is easy to split it to $P_\cX\subset\cX$ and $P_\cY\subset\cY$ such that $P=P_\cX\cup P_\cY$.

\begin{algorithm}
	\caption{Coupled Alpha Complex Filtration Values}
	\label{alg:coupled_alpha_filt}
	\SetKwInOut{Input}{input}\SetKwInOut{Output}{output}
	\Input{$\cX,\cY\subset\Rd$ - finite sets. \\
	$\{S_1,...,S_d\}$ - $S_i$ is the set of $i$-dimensional simplexes of $\cA^\mathrm{co}_{\infty}(\cX,\cY)$.\\}
	\Output{$filtration$ - the filtration values of the simplexes of $\cA^\mathrm{co}_\infty(\cX,\cY)$.}
    \BlankLine
	set all $filtration$ values to NAN\;
	\For{ $k$: $d+1$ to $0$}{
     	\For{all $P$ in $S_k$}{
        		\If{$filtration(P)$ is NAN}{
				$filtration(P)\leftarrow$\textbf{filtrationValue}$(P)$\;
        		}
        		\For{all faces $Q$ of $P$}{
            	    \eIf{$filtration(Q)$ is not NAN}{
                        $filtration(Q)=\min(filtration(Q),filtration(P))$\;}
                    {
                        \If{$Q$ is not coupled Gabriel for $P$}{
                	        $filtration(Q)=filtration(P)$\;
                        }
                    }            
                }
        }    
    }
    \textbf{return} $filtration$\;
\end{algorithm}

\begin{remark}
In \cite{blaser_relative_2020} the authors present the \emph{relative Delaunay-\cech complex}, $\rm{Del}\check{C}(X,A)$, in order to compute the relative-persistent homology $H_k(X,A)$, where $A\subset X\subset\Rd$. This complex is defined by similar building blocks as ours, computing $\cA^\mathrm{co}_\infty(\cX,\cY)$ is identical to the first step in computing all possible simplexes of  $\rm{Del}\check{C}(\cX\cup\cY,\cY)$. However, the filtration values are different. For $\rm{Del}\check{C}(\cX\cup\cY,\cY)$, the filtration value of a simplex $Q\not\subset \cY$ is the radius of the minimal bounding sphere of $Q$ (as in the \cech complex), while the filtration values of all simplexes $Q\subset \cY$ are set to $0$. Thus, making all cycles generated by chains in $C_\bullet (\cY)$ trivial.
\end{remark}

\section{The Number of Simplexes in a Random Coupled Alpha Complex}
\label{sec:num_simp}
In the general (non-random) case, the number of simplexes in a Delaunay triangulation in $(d+1)$-dimensions is $O(n^{\lceil (d+1)/2 \rceil})$ \cite{seidel_upperbound_1995}. Since the coupled alpha complex is a subset of the Delaunay triangulation, it contains at most $O(n^{\lceil(d+1)/2\rceil})$ $(d+1)$-simplexes. However, for a generic random data set the Delaunay triangulation contains $\Theta(n)$ simplexes \cite{dwyer_voronoi_linear_1991}. Unfortunately, this bound cannot be applied directly to the coupled alpha complex, since the data points are restricted to two parallel hyperplanes.

In this section we provide an upper probabilistic bound for the size of the coupled alpha complex, when the points are generated by a Poisson point process.

Let $\Omega\subset\Rd$ be a compact set and let $X_1,X_2,...$ be a sequence of i.i.d.~random variables uniformly distributed in $\Omega$. Let $N\sim \mathrm{Poisson}(n)$ be a Poisson random variable independent of the $X_i$-s. Then, we define  
$$\cP_n=\{X_1,...,X_N\}.$$ 
We say that $\cP_n$ is a homogeneous Poisson point process with intensity $n$.

Let $\Omega\subset\Rd$ be a compact set, and let $\cX_n$ and $\cY_n$ be two independent homogeneous Poisson processes on $\Omega$ with intensity $\lambda=n$. Denote by $F_k$ the number of $k$-simplexes in $\cA^\mathrm{co}_\infty(\cX_n,\cY_n)$. Our main result states that the expected value of $F_k$ is of order $n$. In addition, let $\diam(\Omega)<\rho<\infty$ and assume $\alpha>0$, where
$$\alpha=\min_{x\in\Omega, r\in\mathbb{R}}\frac{\Vol\left(B_r(x)\cap\Omega\right)}{\Vol(B_r(x))}.$$

\begin{prop}
\label{prop:simp_num} 
If $ n\rightarrow\infty$, then the expected number of $k$-simplexes in $\cA^\mathrm{co}_\rho(\cX_n,\cY_n)$ is 
$$\mathbb{E}\{F_k\}=\Theta(n).$$
I.e., 
$$C n\leq F_k \leq D n,$$
for $n\rightarrow\infty$, where $C,D$ are constants that do not depend on $n$. 
\end{prop}

\begin{proof} In order to prove \ref{prop:simp_num}, we first characterize the $(d+1)$-simplexes in $\cA^{\mathrm{co}}$ in spherical-like coordinates. We use the change of variables to prove that $\E\{F_{d+1}\}=O(n)$. Then, we bound the number of $k$-simplexes by 
$$\E\{F_k\}\leq \binom{d+2}{k+1}\E\{F_{d+1}\}.$$

We start by defining $F_{d+1}$ explicitly.
Let $P\subset\cX_n\cup{{\cY}_n}$. Denote $P_\cX=P\cap{{\cX}_n},P_\cY=P\cap{{\cY}_n}$. Assume that $|P_\cX|+|P_\cY|=d+2$ and $P_\cX=l$ for $0 < l < d+2$ (both sets non-empty).
First, we define the following,
$$C(P_\cX,P_\cY) := \bigcap_{x\in P_{\cX}}\vor(x, P_{\cX})\bigcap_{y\in P_{\cY}}\vor(y,P_{\cY}),$$
$$R_\cX(P_\cX,P_\cY):=||C(P_\cX,P_\cY)-x_0||,$$
$$R_\cY(P_\cX,P_\cY):=||C(P_\cX,P_\cY)-y_0||,$$
$$B_\cX(P_\cX,P_\cY):=\big\{z\in\Rd: ||z-C(P_\cX,P_\cY)||\leq R_\cX(P_\cX,P_\cY)\big\},$$
$$B_\cY(P_\cX,P_\cY):=\big\{z\in\Rd: ||z-C(P_\cX,P_\cY)||\leq R_\cY(P_\cX,P_\cY)\big\},$$
where $x_0\in P_\cX$ and $y_0\in P_\cY$ are arbitrary points.
Next, define
$$h_r(P_\cX,P_\cY)=\ind\left\{R_\cX(P_\cX,P_\cY)\leq r\text{ and } R_\cY(P_\cX,P_\cY)\leq r \right\},$$
and
\begin{equation}
\label{eq:g_mlk}
g_{r}(P_\cX,P_\cY; {{\cX}_n}, {{\cY}_n}) = h_r(P_\cX,P_\cY)
\ind\{ B_\cX(P_\cX, P_\cY)\cap {{\cX}_n} = \emptyset \}\ind\{ B_\cY(P_\cX, P_\cY)\cap {{\cY}_n} = \emptyset \}.
\end{equation} 
Finally, we can express $F_{d+1}$ in terms of the above functions, 
\begin{equation} \label{eq: 2}
F_{d+1}
=
	\sum\limits_{l=1}^{d+1} 
		\sum\limits_{\substack{P_\cX\subset {{\cX}_n} \\ |P_\cX|=l}}
			\sum\limits_{\substack{P_\cY\subset {{\cY}_n} \\ |P_\cY|=d-l+2}}
				g_{r}(P_\cX,P_\cY; {{\cX}_n}, {{\cY}_n}).
\end{equation}

Recall that in our case $X=\Rd$, and ${{\cX}_n}$ and ${{\cY}_n}$ are two independent homogeneous Poisson processes with intensity $\lambda=n$. Applying Mecke's formula (see Theorem \ref{th:Palm} in the Appendix A) twice for \eqref{eq: 2} yields,
\begin{equation}\label{eq:last_exp} 
\E\{F_{d+1}\} 
= \frac{n^{d+2}}{(d+2)!}	
	\sum\limits_{l=1}^{d+1} 
		\binom{d+2}{l}	
		\E
		\{
			g_{r}(P_\cX',P_\cY'; {{\cX}_n}\cup P_\cX', {{\cY}_n}\cup P_\cY')
		\}.
\end{equation}
where $P'_\cX$ and $P'_\cY$ are sets of uniform random variables in $\Omega$ of size $l$ and $d-l+2$, respectively, and are independent of ${{\cX}_n}$ and ${{\cY}_n}$.

Conditioning on $(P_\cX',P_\cY')$, and using the fact that ${{\cX}_n}$ and ${{\cY}_n}$ are independent, yields
\begin{multline}
\E
\big\{
\ind\{ B_\cX(P_\cX',P_\cY')\cap {{\cX}_n} = \emptyset \}\ind\{ B_\cY(P_\cX',P_\cY')\cap {{\cY}_n} = \emptyset \}
\ | P_\cX',P_\cY'
\big\}
\\
=
e^{-n\left[\Vol\left(B_\cX(P_\cX',P_\cY')\right)+\Vol\left(B_\cY(P_\cX',P_\cY')\right)\right]}.
\end{multline}
Therefore, using \eqref{eq:g_mlk}, we have
\begin{equation}\label{eq: 6}
\E
\{
	g_{r}(P_\cX',P_\cY'; {{\cX}_n}\cup P_\cX', {{\cY}_n}\cup P_\cY')
\}
=
\E
\left\{ h_r(P_\cX,P_\cY)
e^{-n\left[\Vol\left(B_\cX(P_\cX',P_\cY')\right)+\Vol\left(B_\cY(P_\cX',P_\cY')\right)\right]}
\right\}.
\end{equation}

Next, we compute \eqref{eq: 6} in terms of integrals. 
We start by defining a transformation that takes the vertices of a $(d+1)$-simplex in $\Rd$ to a coordinate system where they are expressed in terms of the minimal pair of open balls defined in Section \ref{sec: filt_vals}.

Let $\bs{x}\in(\Rd)^l, \bs{y}\in(\Rd)^{d+2-l}$ for $0<l<d+2$. Let $s:(\Rd)^{d+2}\rightarrow \Rd\times(\mathbb{R})^2\times (S_{d-1})^{d+2}$ be the map defined by
\begin{equation} \label{eq: var_change}
(\bs{x},\bs{y})\mapsto (z+r_{\bs{x}}\bs{u},z+r_{\bs{y}}\bs{v}),
\end{equation}
where $z=\vor(\bs{x},\bs{x})\cap\vor(\bs{y},\bs{y})$, 
$r_{\bs{x}}=\|x-z\|,r_{\bs{y}}=\|y-z\|$ for arbitrary $x\in\bs{x},y\in\bs{y}$, and $\bs{u},\bs{v}$ are points on the $(d-1)$-dimensional unit sphere. First note that the number of degrees of freedom in the codomain equals $d+2+(d-1)(d+2)=d(d+2)$, which is equal to that of the domain. The point $z$ is a unique point (see Lemma \ref{lemma:plns_int_dim} in the Appendix A), hence the map $s$ is a bijection. 
Note that for $r\geq\diam(\Omega)$ we have
$$\Omega\subset\bigcap_{x\in\cX_n\cup\cY_n}B_r(x),$$
which implies that the homology does not change if $r$ is increased.
Hence, we can limit our analysis for simplexes with $r\leq\rho=\diam(\Omega)$ and $z\in\Omega_\rho:=\cup_{x\in\Omega}B_\rho(x)$.
We can then bound \eqref{eq: 6} in the following way,
\begin{equation*}
\int\limits_{(\bs{x},\bs{y})\subset\Omega^{d+2}} g_{\rho}(\bs{x},\bs{y})d\bs{x}d\bs{y}
\end{equation*}
\begin{equation} \label{eq: first_bound}
\leq
\int\limits_{z\in\Omega_\rho}
	\int\limits_{r_{\bs{x}}\geq 0}^{\rho}
		\int\limits_{r_{\bs{y}}\geq 0}^{\rho}
			\int\limits_{\bs{u}\in S^{l}}
					\int\limits_{\bs{v}\in S^{d-l+2}}
e^{-n(\alpha\omega_d r_{\bs{x}}^d+\alpha\omega_d r_{\bs{y}}^d)}					
J(z,r_{\bs{x}},r_{\bs{y}},\bs{u},\bs{v})
					d\bs{v}
			d\bs{u}
		dr_{\bs{y}}
	dr_{\bs{x}}
dz,
\end{equation}
where $\omega_d$ is the volume of the unit ball in $\Rd$.

From \eqref{eq: var_change}, we have that the Jacobian satisfies $J(z,r_{\bs{x}},r_{\bs{y}},\bs{u},\bs{v})=(r_{\bs{x}}^{d-1})^l (r_{\bs{y}}^{d-1})^{d-l+2}J(0,1,1,\bs{u},\bs{v}).$
Hence, we can rewrite the RHS of \eqref{eq: first_bound} in the following way,
$$
	\int\limits_{0}^{\rho}
		r_{\bs{x}}^{d\left(\frac{l(d-1)+1}{d}\right)-1} e^{-n\alpha\omega_d r_{\bs{x}}^d}
		dr_{\bs{x}}
	\int\limits_{0}^{\rho}
		r_{\bs{y}}^{d\left(\frac{d(d-l+1)+l-1}{d}\right)-1}
		e^{-n\alpha\omega_d r_{\bs{y}}^d}
		dr_{\bs{y}}
\int\limits_{z\in\Omega_\rho} dz
\int\limits_{\bs{u}\in S^{l}} 
\int\limits_{\bs{v}\in S^{d-l+2}}
J(0,1,1,\bs{u},\bs{v})
d\bs{v} d\bs{u}
.
$$
Finally, using the changes of variables apply the following change of variables $n\alpha\omega_d r_x^d\rightarrow t_x$, $n\alpha\omega_d r_y^d\rightarrow ty$, and the definition of the incomplete gamma function $\gamma(\cdot, \cdot)$, we have
\begin{align}\label{eq:last_int}
\begin{split} 
\int\limits_{(\bs{x},\bs{y})\subset\Omega^{d+2}} g_{\rho}(\bs{x},\bs{y})d\bs{x}d\bs{y}
& \leq
\frac{|\Omega_\rho|}{d^2(n\alpha\omega_d)^{d+1}}
C_{l}^{d}
\gamma\left(\frac{l(d-1)+1}{d},\alpha\Lambda\right)\gamma\left(\frac{d(d-l+1)+l-1}{d},\alpha\Lambda\right)
\\
&  \leq
\frac{|\Omega_\rho|}{d^2(n\alpha\omega_d)^{d+1}}
C_{l}^{d}
\gamma^2(d+1,\alpha\Lambda),
\end{split}
\end{align}
where 
$$\Lambda := n\omega_d{\rho}^d,\quad\text{and}\quad C_{l}^{d}
:=
	\int\limits_{\bs{u}\in S^{l}}
		\int\limits_{\bs{v}\in S^{d-l+2}}
			J(0,1,1,\bs{u},\bs{v})
		d\bs{v}
	d\bs{u}.
$$
To conclude, substituting \eqref{eq:last_int} into \eqref{eq:last_exp} we have that the number of $(d+1)$-simplexes is bounded by
\begin{equation*}
\E\{F_{d+1}\}
\leq
\frac{n|\Omega_\rho|}{d^2(\alpha\omega_d)^{d+1}} 
	\frac{\gamma^2(d+1,\alpha\Lambda)}{\Gamma(d+1)}
	\sum\limits_{l=1}^{d+1} \binom{d+2}{l} C_{l}^{d}\\
= O(n)
\end{equation*}
where we used the fact that $\gamma(d+1,\cdot)\leq \Gamma(d+1)$. 
Hence, we conclude that the total number of $k$-simplexes centered at $\Omega_\rho$ in $\coA_{\infty}({{\cX}_n},{{\cY}_n})$ is bounded from above by
\begin{align*}
\begin{split}
\E\{F_{k}\}
\leq
\binom{d+2}{k+1} \E\{F_{d+1}\} = O(n).
\end{split}
\end{align*}
The following result given in \cite{edelsbrunner_expected_2017} is the number of $k$-simplexes in the alpha complex (when $r\rightarrow\infty$) of a Poisson point process with intensity $n$, that are centered at 
$\Omega_\rho$,
\begin{equation} \label{eq:edel_alpha}
n|\Omega_\rho|\sum_{m=k}^d\sum_{l=0}^k\binom{m-l}{m-k}\tilde{C}^d_{k,m}.
\end{equation}
By using \eqref{eq:edel_alpha}, we can bound the number of simplexes from below by the number of simplexes in  $\cA_{\infty}({{\cX}_n}\cup{{\cY}_n})$, since $\cA_{\infty}({{\cX}_n}\cup{{\cY}_n})\subset\coA_{\infty}({{\cX}_n},{{\cY}_n})$. 
Hence,  we conclude that 
$$\E\{F_k\}
=
\Theta(n)
.
$$
\end{proof}

\appendix
\section*{Appendix}
\label{sec: appx}
\section{Meck's Formula} 
The following theorem is a result of Palm theory for Poisson point processes.
\begin{thm}[Mecke's formula] \label{th:Palm}
Let $(X,\rho)$ be a metric space, let $f:X\rightarrow\mathbb{R}$ be a probability density on $X$, and let $\Pn$ be a random Poisson process on $X$ with intensity $\lambda_n=nf$. Let $h(\cY,\cX)$ be a measurable function defined for all finite subsets $\cY\subset\Pn\subset X$ with $|\cY|=k$. Then
\begin{equation} \label{eq:Mecke}
\E\left\{ \sum\limits_{\substack{\cY\subset \Pn \\ |\cY|=k}}h(\cY,\Pn) \right\}=\dfrac{n^k}{k!}\E\{h(\cY',\cY'\cup\Pn)\},
\end{equation}
where the sum is over all subsets $\cY\subset\Pn$ of size $|\cY|=k$, and $\cY'$ is a set of $k$ iid random variables in $X$ with density $f$, independent of $\Pn$.
\end{thm}
For a proof of Theorem \ref{th:Palm} see \cite{penrose_random_2003}.

\section{Random Hyperplanes Lemma}
Let $Q,P\subset\Rd$ be two subsets, such that $|Q|=l$ and $|P|=d-l+2$, for $0 < l < d+2$.
Denote 
$$K_Q:=\bcap_{q\in Q}\vor(q,Q),\quad\text{and}\quad K_P:=\bcap_{p\in P}\vor(p,P).$$
In other words, $K_Q$ and $K_P$ are all the points that equidistant from the points $Q$ and $P$ respectively.

\begin{lem} \label{lemma:plns_int_dim}
Let $P,Q$ be sets of iid  points sampled from a distribution with a density in $\Rd$. The following holds almost surely.
$$\dim(K_Q)=d-l+1,\quad\dim(K_P)=l-1,$$ 
and
$$\dim(K_Q\cap K_P) = 0.$$
\end{lem}

\begin{proof} Assume an arbitrary ordering on $Q$ and $P$. The intersection $K_Q\cap K_P$ can be expressed as the set of all points $x\in\Rd$ that solve the system, $Ax=b$, where
\[
A=
\begin{bmatrix}
(q_2-q_1)^T	\\
\vdots \\
(q_l-q_1)^T	\\
(p_2-p_1)^T	 \\
\vdots \\
(p_{d-l+2}-p_1)^T 
\end{bmatrix}
,\ x=
\begin{bmatrix}
x_1 \\
\vdots \\
x_d
\end{bmatrix}
,\ b= 
\frac{1}{2}
\begin{bmatrix}
||q_2||^2-||q_1||^2\\
\vdots \\
||q_l||^2-||q_1||^2 \\
||p_2||^2-||p_1||^2\\
\vdots \\
||p_{d-l+2}||^2-||p_1||^2 
\end{bmatrix}.
\] 
Since the points of $Q\cup P$ are iid , with a density in $\Rd$, they are almost surely in a coupled general position (Definition \ref{def:coup_gen_pos}). Hence, $\rank(A)=d$ (since its rows are linearly independent) and there is a unique solution to the system, i.e.~a point. Note that if we apply similar arguments to the part of the matrix that contains $Q$ only or $P$ only, we get $\dim(K_Q)=d-l+1$ and $\dim(K_P)=l-1$.  
\end{proof}

\newpage

\vskip 0.2in
\bibliographystyle{plain}
\bibliography{general}

\end{document}